\documentclass[a4paper,USenglish,cleveref,autoref,thm-restate,pdfa]{lipics-v2021}
\nolinenumbers 
\hideLIPIcs

\usepackage{tcolorbox}
\usepackage{xcolor}
\usepackage{xspace}
\usepackage{xargs}

\title{The Parameterized Complexity of Coloring~Mixed~Graphs}

\author{Antonio Lauerbach}
{Julius-Maximilians-Universität Würzburg, Germany \and \url{https://go.uniwue.de/lauerbach-antonio}}
{lauerbach@informatik.uni-wuerzburg.de}
{https://orcid.org/0009-0007-9093-3443}
{}
\author{Konstanty Junosza-Szaniawski}{Warsaw University of Technology, Poland}{}{https://orcid.org/0000-0003-0352-8583}{}
\author{{Marie Diana} Sieper}{Julius-Maximilians-Universität Würzburg, Germany}{}{https://orcid.org/0009-0003-7491-2811}{}
\author{Alexander~Wolff}{Julius-Maximilians-Universität Würzburg, Germany \and \url{https://www.informatik.uni-wuerzburg.de/en/algo/team/wolff-alexander/}}{}{https://orcid.org/0000-0001-5872-718X}{}

\authorrunning{A.~Lauerbach, K.~Junosza-Szaniawski, M.\,D.~Sieper, and A.~Wolff}
\Copyright{Antonio Lauerbach, Konstanty Junosza-Szaniawski, Marie Diana Sieper, Alexander Wolff}


\acknowledgements{This work was started at the workshop Homonolo 2024
  in Nov\'a Louka.  We thank the organizers and the other
  participants.}

\graphicspath{{figures/}}

\newcommand{\arxiv}[2]{#1}

\usepackage{tikz}
\usetikzlibrary{arrows,math,patterns,shapes,fit,decorations.pathreplacing,shapes.geometric}

\usepackage{optidef}

\newcommand{\repeatcaption}[2]{%
  \renewcommand{\thefigure}{\ref{#1}}%
  \captionsetup{list=no}%
  \caption{#2}
  \addtocounter{figure}{-1}
}

\newif\ifinappendix
\let\oldappendix\appendix
\renewcommand{\appendix}{
  \oldappendix
  \inappendixtrue
}

\newcommand{\restateref}[1]{\ifinappendix{\hyperref[#1]{$\star$}}\else{\hyperref[#1*]{$\star$}}\fi}

\usepackage{nicefrac}

\usepackage{booktabs}

\DeclareRobustCommand{\abbrevcrefs}{%
\crefname{theorem}{Thm.}{Thms.}%
\crefname{corollary}{Cor.}{Cors.}%
}
\DeclareRobustCommand{\cshref}[1]{{\abbrevcrefs\cref{#1}}}

\definecolor{defblue}{rgb}{0.121,0.47,0.705}
\definecolor{linkblue}{rgb}{0.098,0.098,0.4392}
\let\emph\relax
\DeclareTextFontCommand{\emph}{\color{defblue}\em}

\let\leq\leqslant
\let\geq\geqslant
\let\le\leqslant
\let\ge\geqslant
\let\epsi\varepsilon
\let\rho\varrho

\newcommand{\ceil}[1]{{\left\lceil #1 \right\rceil}}

\newtheorem{probspecint}{Problem}

\newcommand{\defproblem}[4]{
  \begin{tcolorbox}%
    \hspace{0ex}\hspace*{-2.8ex}
    \begin{minipage}{0.99\textwidth}
      \vspace{0ex}\vspace*{-1ex}
      \begin{tabular}{@{}l@{~~}p{0.9\textwidth}@{}}
        {\sf\bfseries\color{gray}\normalsize Problem:} & \normalsize#1\\[.1ex]
        {\sf\bfseries\color{gray}\normalsize Input:} & \normalsize#2\\[.1ex]
        {\sf\bfseries\color{gray}\normalsize #4:} & \normalsize#3\\[-1ex]
      \end{tabular}
    \end{minipage}
  \end{tcolorbox}
}
\newcommand{\defdecproblem}[3]{\defproblem{#1}{#2}{#3}{Question}}

\newcommand{\NP}{\ensuremath{\mathsf{NP}}\xspace}

\newcommand{\FPT}{\ensuremath{\mathsf{FPT}}\xspace}
\newcommand{\XP}{\ensuremath{\mathsf{XP}}\xspace}
\newcommand{\W}[1]{\ensuremath{\mathsf{W[#1]}}\xspace}

\newcommand{\PMulticoloredClique}{\textsc{MulticoloredClique}\xspace}
\newcommand{\PShortestSuperstring}{\textsc{ShortestSuperstring}\xspace}
\newcommand{\PPrecedenceConstrainedScheduling}{\textsc{PrecedenceConstrainedScheduling}\xspace}

\usepackage{algorithm2e}
\SetKwProg{Fn}{}{:}{} 
\SetKwProg{Alg}{Algorithm}{}{} 

\DeclareMathOperator{\ag}{AG} 

\usepackage{mathtools}
\DeclarePairedDelimiter\set{\{}{\}} 
\DeclarePairedDelimiter\abs{\lvert}{\rvert} 
\DeclarePairedDelimiter\floor{\lfloor}{\rfloor} 
\DeclarePairedDelimiter\angles{\langle}{\rangle} 

\def\to{\ensuremath{\rightarrow}} 
\def\Oh{\ensuremath{\mathcal{O}}} 

\newcommand{\Poly}{\textsf{P}}
\newcommand{\paraNP}{\textsf{paraNP}}

\newcommand{\N}{\mathbb{N}} 

\newcommand{\SETH}{\textsf{SETH}}

\newcommand{\XColoring}[1]{\textsc{#1Coloring}}
\newcommand{\Coloring}{\XColoring{}}
\newcommand{\MixedColoring}{\XColoring{Mixed}}
\newcommand{\ListColoring}{\XColoring{List}}

\newcommand{\SAT}{\textsc{SAT}}
\newcommand{\ILP}{\textsc{ILP}}

\newcommand{\cw}{\ensuremath{\mathsf{cw}}} 
\newcommand{\tw}{\ensuremath{\mathsf{tw}}} 
\newcommand{\pw}{\ensuremath{\mathsf{pw}}} 
\newcommand{\td}{\ensuremath{\mathsf{td}}} 
\newcommand{\vc}{\ensuremath{\mathsf{vc}}} 
\newcommand{\fvs}{\ensuremath{\mathsf{fvs}}} 
\newcommand{\ndu}{\ensuremath{\mathsf{nd}}} 
\newcommand{\ndm}{\ensuremath{\mathsf{mnd}}} 
\newcommand{\cwu}{\ensuremath{\mathsf{cw}}} 
\newcommand{\cwm}{\ensuremath{\mathsf{mcw}}} 
\newcommand{\cwd}{\ensuremath{\mathsf{dcw}}} 

\newcommand{\maxrank}{\ensuremath{\Lambda}}

\newcommand{\inrank}{\ensuremath{\lambda^-}}

\newcommand{\neigh}{\ensuremath{{N}}}
\newcommand{\inneigh}{\ensuremath{{N}^-}}
\newcommand{\outneigh}{\ensuremath{{N}^+}}
\newcommand{\uneigh}{\ensuremath{{N}^{\mathrm{u}}}}

\newcommand{\cwedge}{\eta}
\newcommand{\cwarc}{\alpha}
\newcommand{\cwunion}{\oplus}
\newcommand{\cwrelabel}{\rho}

\begin{document}

\maketitle

\begin{abstract}
  A mixed graph contains (undirected) edges as well as
  (directed) arcs, thus generalizing undirected and directed graphs.
  A proper coloring~$c$ of a mixed graph~$G$ assigns a positive
  integer to each vertex such that $c(u)\neq c(v)$ for every
  edge~$\set{u,v}$ and~$c(u)<c(v)$ for every arc~$(u,v)$ of~$G$. 
  As in classical coloring, the objective is to minimize the number of colors.
  Thus, mixed (graph) coloring generalizes classical coloring of
  undirected graphs and allows for more general applications, such as
  scheduling with precedence constraints, modeling metabolic pathways,
  and process management in operating systems; see a survey by Sotskov
  [Mathematics, 2020].

  We initiate the systematic study of the parameterized complexity of
  mixed coloring.  We focus on structural graph parameters that lie
  between cliquewidth and vertex cover, primarily with respect to the
  underlying undirected graph.  Unlike classical coloring, which is
  fixed-parameter tractable (\FPT) parameterized by treewidth or
  neighborhood diversity, we show that mixed coloring is \W{1}-hard
  for treewidth and even \paraNP-hard for neighborhood diversity.  To
  utilize the directedness of arcs, we introduce and analyze natural
  generalizations of neighborhood diversity and cliquewidth to mixed
  graphs, and show that mixed coloring becomes \FPT{} when
  parameterized by (the generalized) mixed neighborhood diversity.  Further, we
  investigate how these parameters are affected if we add transitive
  arcs, which do not affect colorings.  Finally, we provide tight bounds on
  the chromatic number of mixed graphs, generalizing known bounds on
  mixed interval graphs.
\end{abstract}

\section{Introduction}
\label{sec:introduction}
The problem of coloring graphs, i.e., assigning as few colors as possible to vertices of an undirected graph such that adjacent vertices receive distinct colors, is an old and extensively studied problem in graph theory with applications in various areas, such as scheduling, frequency assignment, and graph drawing~\cite{Kubale04GraphColorings}. 
An example of a scheduling problem that can be modeled using graph coloring is timetabling, where courses (vertices) need to be scheduled in certain time slots (colors) while respecting conflicts (edges), e.g., due to two courses having the same teacher or students. These types of problems are referred to as chromatic scheduling problems. In real world applications, there are usually several additional constraints, such as certain courses having to be scheduled in specific time slots, certain teachers not being available at certain times, courses using multiple consecutive time slots, or a lecture having to be held before the corresponding exercises. Thus, there exist several extensions of coloring to accommodate these additional constraints, such as using list coloring to restrict courses to certain slots~\cite{Werra97RestrictedColorModelsTimetabling}, or using interval coloring to enforce consecutive time slots~\cite{Kubale89IntervalVertexColoring}.
To account for precedence constraints, for example, if a course has to
be scheduled before another, Hansen et
al.~\cite{HansenKdW97MixGraphColor} added arcs to the graph, with the
idea that the arc~$(u,v)$ expresses that~$u$ has to be scheduled
before~$v$.
This leads to the notion of a \emph{mixed graph}~$G$, which consists of a set \emph{$V(G)$} of \emph{vertices}, a set \emph{$E(G)$} of (undirected) \emph{edges}, and a set \emph{$A(G)$} of (directed) \emph{arcs}.  A mixed graph combines the concepts of undirected and directed graphs. 
We require mixed graphs to be \emph{simple}, i.e., without loops and without parallel edges or arcs.

A \emph{$k$-coloring}~$c \colon V(G)\to\set{1,\dots,k}$ of a mixed graph~$G$ is \emph{proper} if it holds for every edge $\set{u,v}\in E(G)$ that~$c(u)\neq c(v)$ and it holds for every arc~$(u,v)\in A(G)$ that~$c(u)<c(v)$.
We call the problem of deciding whether a mixed graph can be properly colored with~$k$ colors {\color{defblue}$k$-\MixedColoring{}}. 
As a mixed graph cannot be properly colored if it contains a directed cycle, we consider only mixed graphs without directed cycles.
Apart from scheduling, mixed graph colorings have also been applied in graph drawing in order to compact layered orthogonal drawings~\cite{GutowskiMRSWZ22ColoringMixedDirIntGraph}. See the survey by Sotskov~\cite{Sotskov20SurveyMixGraphColor} for more details and applications of mixed graph colorings.

Since \MixedColoring{} generalizes the classical \Coloring{} problem for undirected graphs, it is also \NP-hard for three or more colors. 
As such, it is natural to study its parameterized complexity. 
However, to our knowledge, there has not yet been a systematic study of the parameterized complexity of \MixedColoring{}. 
The only known results in this direction are an \FPT-algorithm
parameterized by the number of edges (i.e., not counting arcs) by Hansen et al.~\cite{HansenKdW97MixGraphColor}, with subsequent improvements by Damaschke~\cite{Damascke19ParaMixGraphColor}, as well as an \XP-algorithm parameterized by treewidth by Ries and de Werra~\cite{Ries-deWerra08TwoColorProbMixGraph}.

Meanwhile, the parameterized complexity of \Coloring{} has been extensively studied.
On the positive side, results include the well-known dynamic program that runs in~$\Oh^*(k^\tw)$ time, where~$k$ is the number of colors, $\tw$ is the treewidth of the graph, and the $\Oh^*$-notation hides polynomial factors. Since treewidth bounds the chromatic number of undirected graphs, this implies that \Coloring{} is \FPT{} parameterized by treewidth. Lokshtanov et al.~\cite{LokshtanovMS18KnownAlgoTWProbOpt} even showed that under \SETH{} (a hypothesis about the complexity of \SAT{} with at most~$k$ literals per clause) there is no~$\Oh^*((k-\epsi)^\tw)$-time algorithm for any~$\epsi>0$, i.e., the~$\Oh^*(k^\tw)$ algorithm is probably optimal. Further results include \Coloring{} being \FPT{} parameterized by neighborhood diversity, as shown by Ganian~\cite{Ganian12NeighborDiversityColoringFPT}.
On the negative side, Fomin et al.~\cite{FominGLS10IntractabilityCliqueWidth} showed that \Coloring{} is \W{1}-hard parameterized by cliquewidth. However, \Coloring{} still admits an \XP-algorithm parameterized by cliquewidth, as shown by Kobler and Rotics~\cite{KoblerRotics03XPColoringCliqueWidth}. Further, applying Courcelle's theorem~\cite{Courcelle90Theorem,CourcelleEngelfriet12GSaMSOL}, a famous meta-theorem that yields \FPT-algorithms, it follows that \Coloring{} is \FPT parameterized by cliquewidth plus chromatic number.\looseness=-1

Understanding the relationships between various structural graph parameters is crucial for analyzing the parameterized complexity of a problem. We say that a graph parameter~$\alpha$ \emph{(upper) bounds} another parameter~$\beta$ if there exists a computable function~$f$ such that for every (mixed) graph~$G$, it holds that~$\beta(G)\leq f(\alpha(G))$. If no such function exists, we say that~$\alpha$ \emph{does not (upper) bound}~$\beta$. If~$\alpha$ (upper) bounds~$\beta$ and~$\beta$ does not (upper) bound~$\alpha$, we say that~$\alpha$ \emph{strictly (upper) bounds}~$\beta$. Meanwhile, if neither parameter bounds the other, we say that they are \emph{incomparable}. Conversely, if two parameters bound each other, we say that they are \emph{equivalent}.
For example, since treewidth is bounded by vertex cover, this implies that any problem that is \FPT parameterized by treewidth, such as \Coloring{}, is also \FPT parameterized by vertex cover. 

For mixed graphs however, it is not immediately clear how to apply the existing parameters, as they are defined for undirected graphs.
The approach taken in previous works is to consider the parameters w.r.t.\ the \emph{underlying undirected graph}, i.e., the graph obtained by replacing all arcs with edges.
The downside of this approach is that it neglects the directedness of arcs. In fact, Courcelle and Olariu~\cite{CourcelleOlariu00BoundsCliqueWidth} originally defined cliquewidth for undirected as well as for directed graphs, yielding a natural generalization for mixed graphs.
Further, Fernau et al.~\cite{FernauFMPN25ParaPathPartitions}
generalized neighborhood diversity to directed graphs by
distinguishing between incoming and outgoing neighbors; this
can also be generalized to mixed graphs.

\begin{figure}[tb]
  \centering
  \begin{minipage}[b]{0.57\textwidth}
    \centering
    \includegraphics{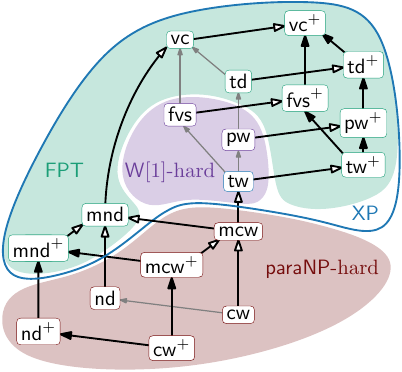}
  \end{minipage}
  \hfill
  \begin{minipage}[b]{0.39\textwidth}
    \centering
    \begin{tabular}{@{}lcr@{}}
      \toprule
      Parameter & Complexity & Reference \\
      \midrule
      $\cwm+\chi$ & \FPT & \cshref{thm:logic-fpt-cw-k}\\
      $\ndm$ & \FPT & \cshref{thm:fpt-ndm}\\
      $\ndu+\chi$ & \FPT & \cshref{thm:fpt-ndu-chrom}\\
      $\tw$ & \XP & Cor.\hspace*{1.9ex}\ref{cor:tw-fpt-xp}\\
      \midrule
      $\fvs+\pw$ & \W{1}-hard & \cshref{thm:w1-fvs-pw}\\
      $\cwm+\maxrank$ & \paraNP-hard & \cshref{thm:para-np-cwm-rank}\\
      $\ndu^+$ & \paraNP-hard & \cshref{thm:paraNP-ndu}\\
      \bottomrule
    \end{tabular}
    \vspace*{1ex}
  \end{minipage}
  
  \begin{minipage}[t]{0.57\textwidth}
    \captionof{figure}{Overview of the hierarchy of structural
      parameters of mixed graphs (without directed cycles), and the
      parameterized complexity landscape for \MixedColoring{}.  Arrows
      represent strict (upper) bounds between parameters.  Bold arrows
      indicate new relationships established in this work.
      White-tipped arrows highlight results that are not
      straightforward consequences of known results.  Composite
      parameters are omitted for clarity.}
    \label{fig:mix-para-hierarchy}
    \label{fig:para-complexity}
  \end{minipage}
  \hfill
  \begin{minipage}[t]{0.39\textwidth}  
    \captionof{table}{Our parameterized complexity results for
      \MixedColoring{} including results for composite parameters. The
      complexity of other parameters can be inferred from the
      parameter hierarchy in \cref{fig:mix-para-hierarchy}.}
    \label{tab:para-complex-results}
  \end{minipage}
\end{figure}
\vspace{-1.1ex}

\subparagraph{Contribution}
In order to better explore the parameterized complexity of \MixedColoring{}, we first introduce and analyze the generalizations of neighborhood diversity and cliquewidth to mixed graphs, which we call mixed neighborhood diversity ($\ndm$) and mixed cliquewidth ($\cwm$), respectively. Since in \MixedColoring{} we can assume that the input graph does not contain directed cycles, we only analyze these parameters for mixed graphs without directed cycles. 
We explore their relationships with existing parameters defined on the underlying undirected graph, such as vertex cover ($\vc$), treedepth ($\td$), pathwidth ($\pw$), feedback vertex set ($\fvs$), treewidth ($\tw$), as well as neighborhood diversity ($\ndu$) and cliquewidth ($\cwu$); see \arxiv{\cref{appx:parameters}}{the full version~\cite{FullVersion}} for detailed definitions. Specifically, we show that $\cwm$ strictly bounds $\cwu$ (\cref{prop:cwm-cwu}) and is strictly bounded by $\tw$ (\cref{prop:cwm-tw}). Similarly, we show that $\ndm$ strictly bounds $\ndu$ (\cref{prop:ndm-ndu}) and is strictly bounded by $\vc$ (\cref{prop:ndm-vc}). Furthermore, we show that $\ndm$ strictly bounds $\cwm$ (\cref{prop:ndm-cwm}), while~$\cwm$ and~$\ndu$ are incomparable (\cref{prop:cwm-ndu-incomp}).

Since arcs enforce precedence constraints on the colors of the vertices, \emph{transitive arcs}, i.e., arcs between vertices connected by a directed path along other arcs, do not influence the properness of colorings of mixed graphs. Adding all transitive arcs (and removing any edges parallel to transitive arcs) results in the \emph{transitive closure}~{$G^+$} of a mixed graph~$G$ (without directed cycles). 
We denote with~$\alpha^+$ the \emph{transitive closure counterpart} of the parameter~$\alpha$, i.e.,~$\alpha^+(G)=\alpha(G^+)$.
We show that $\vc$, $\td$, $\pw$, $\fvs$, and $\tw$ are strictly bounded by their transitive closure counterparts (\cref{prop:param-tc-upper-bound}), while for $\cwm$ and $\ndm$ the situation is inverted, with $\cwm$ and $\ndm$ strictly bounding $\cwm^+$ and $\ndm^+$, respectively (\cref{prop:ndm/cwm-upper-bound-tc}). Meanwhile, $\cwu$ and $\ndu$ are incomparable to their transitive closure counterparts (\cref{prop:ndu-cwu-tc-incomparable}).
The resulting hierarchy of parameters is visualized in \cref{fig:mix-para-hierarchy}.

Furthermore, we also consider bounds for the \emph{chromatic number}~(\emph{$\chi$}) of mixed graphs, i.e., the minimum number of colors needed for a proper coloring. 
After providing some simple lower bounds, i.e., that the chromatic number of a mixed graph is lower-bounded by the chromatic number of the underlying undirected graph as well as by \emph{maxrank}~(\emph{$\maxrank$}), the length of the longest directed path, we generalize and tighten a bound by Gutowski et al.~\cite{GutowskiJKR0Z23ColorRecMixIntGraph} on the chromatic number of mixed interval graphs (\cref{thm:chi-upper-bound}). Based on this, we show that the chromatic number of a mixed graph is bounded by~$\tw^+$ as well as~$\td$ and~$\vc$ (\cref{cor:chi-bounds}).

Using the newly established parameter relationships, we analyze the parameterized complexity of \MixedColoring{} in \cref{sec:para-complexity}. The results are summarized in \cref{tab:para-complex-results} and visualized in \cref{fig:para-complexity}.
We conclude with an outlook on open problems in \cref{sec:conclusion}.

\subparagraph{Coloring Directed Graphs}
For directed graphs, Neumann-Lara~\cite{Neumann-Lara82DichromaticNumberDigraph} introduced the dichromatic number and the corresponding coloring, where each color class induces an acyclic subgraph.
This problem, known as \textsc{DigraphColoring}~\cite{HarutyunyanEtAl24DigraphColoringDistanceAcyclicity}, also generalizes coloring undirected graphs, as each edge can be replaced by two opposite arcs.
Therefore, \textsc{DigraphColoring} is \NP-hard.
In contrast, \MixedColoring{} is easy on directed graphs: if the graph is acyclic, greedily coloring vertices in topological order yields an optimal coloring.
Otherwise, the graph cannot be properly colored. 
This difference arises due to the fact that \textsc{DigraphColoring}
treats arcs very differently: in \MixedColoring, there
cannot be arcs between vertices of the same color, whereas in
\textsc{DigraphColoring} there can be arcs between vertices of the
same color, as long as these arcs do not form a directed cycle.

On directed acyclic graphs, \MixedColoring{}, especially when
considering the transitive closure, behaves somewhat similar to
\textsc{OrientedColoring}~\cite{Sopena16HomomorphismsColouringsOrientedGraphsUpdatedSurvey},
where, between each pair of color classes, all arcs must have the same
orientation.  However, contrary to \MixedColoring{},
\textsc{OrientedColoring} does not forbid directed cycles.

\subparagraph{Conventions}
For positive integer~$k$, we use~\emph{$[k]$} as shorthand for~$\set{1,\dots,k}$. 
We say that a mixed graph~$G$ is a \emph{partial orientation} of its underlying undirected graph.
The full versions of proofs marked with a (clickable) star~(\restateref{}) can be found in the \arxiv{appendix}{full version of this paper~\cite{FullVersion}}.

\section{Mixed Graph Parameters}
\label{sec:mixed-params}
We define the \emph{mixed cliquewidth} of a mixed graph~$G$, denoted by~\emph{$\cwm(G)$}, to be the minimum number of distinct labels needed to construct~$G$ using the following operations:
\begin{itemize}
  \item creation of a new vertex with label~$i$, denoted by~\emph{$i$}.
  \item disjoint union of two labeled mixed graphs, denoted by~\emph{$\cwunion$}.
  \item adding an (undirected) edge between every vertex with label~$i$ and every vertex with label~$j$, for~$i\neq j$, denoted by~\emph{$\cwedge_{i,j}$}.
  \item adding a (directed) arc from every vertex with label~$i$ to every vertex with label~$j$, for~$i\neq j$, denoted by~\emph{$\cwarc_{i,j}$}.
  \item renaming label~$i$ to label~$j$, denoted by~\emph{$\cwrelabel_{i \to j}$}.
\end{itemize}

These are exactly the operations used by Courcelle and Olariu~\cite{CourcelleOlariu00BoundsCliqueWidth} in their definitions of cliquewidth for undirected and directed graphs. Instead of using only~$\cwedge_{i,j}$ for undirected graphs and only~$\cwarc_{i,j}$ for directed graphs, we allow both operations to be used in order to construct mixed graphs, thereby generalizing the concept of cliquewidth to mixed graphs.
We call the sequence of such operations constructing a given mixed graph using~$\ell$ labels a \emph{mixed $\ell$-expression}. For example, the directed path of length~$4$ is constructed by the mixed $3$-expression $\cwarc_{2,3}(\cwrelabel_{3\to 2}(\cwrelabel_{2\to 1}(\cwarc_{2,3}(\cwarc_{1,2}(1\cwunion 2)\cwunion 3)))\cwunion 3)$.
As it turns out, mixed cliquewidth is not bounded by cliquewidth.

\begin{proposition}\label{prop:cwm-cwu}
  For any mixed graph~$G$ (without directed cycles), $\cwu(G)\leq \cwm(G)$.
  Furthermore,~$\cwu$ does not bound~$\cwm$.
\end{proposition}
\begin{proof}
  The first part follows directly from the definitions, as replacing all arc operations by edge operations in a mixed $\ell$-expression constructing~$G$ yields an $\ell$-expression constructing the underlying undirected graph of~$G$.

  Let~\emph{$\ag(G)$} be the subgraph of~$G$ resulting from omitting all edges. It holds that $\cwm(\ag(G))\leq \cwm(G)$, as by omitting the edge operations in a mixed $\ell$-expression constructing~$G$ we obtain a mixed $\ell$-expression for~$\ag(G)$.
  Using this, we construct a family of mixed graphs~$(G_\ell)_{\ell\geq 1}$ with~$\cwu(G_\ell)=2$ and~$\cwm(G_\ell)\geq \ell+1$.
  We construct~$G_\ell$ by orienting an $(\ell\times \ell)$-grid graph such that no directed cycles are created. As shown by Golumbic and Rotics~\cite{GolumbicRotics00CliquewidthPerfectGraphClasses}, the cliquewidth of an $(\ell\times \ell)$-grid graph is~$\ell+1$ and thus $\cwm(G_\ell)\geq\cwm(\ag(G_\ell))\geq\cwu(\ag(G_\ell))=\ell+1$. Furthermore, between any pair of vertices of~$G_\ell$ that are not connected by an arc, we add an edge. This results in the underlying undirected graph of~$G_\ell$ being a complete graph, which has cliquewidth~$2$, and thus~$\cwu(G_\ell)=2$.\looseness=-1
\end{proof}

Courcelle and Olariu~\cite{CourcelleOlariu00BoundsCliqueWidth} showed that the cliquewidth of undirected as well as of directed graphs is bounded by treewidth. We can utilize their result for directed graphs to show that mixed cliquewidth (on mixed graphs without directed cycles) is also bounded by treewidth. To this end, we transform a mixed graph~$G$ without directed cycles into its \emph{corresponding directed graph}~\emph{$D(G)$} by replacing each edge~$\set{u,v}$ by two opposite arcs~$(u,v)$ and~$(v,u)$. Since mixed graphs without directed cycles cannot contain opposite arcs (as they would form a directed cycle of length two) this transformation is injective. We define the \emph{directed cliquewidth} of a mixed graph~$G$ without directed cycles, denoted~$\cwd(G)$, to be the (directed) cliquewidth of the corresponding directed graph~$D(G)$. As we show now, mixed and directed cliquewidth are, while not equal, equivalent under the lens of parameterized complexity.
\begin{restatable}[\restateref{prop:cwm-tw}]{proposition}{CWMtoTW}
  \label{prop:cwm-tw}
  For any mixed graph~$G$ (without directed cycles), it holds
  that $\cwd(G)\leq\cwm(G)\le 2\,\cwd(G)\leq 2(2^{\tw(G)+1}+1)$.
  Further, there exists a mixed graph~$\tilde{G}$ such
  that~$\cwd(\tilde{G})\neq\cwm(\tilde{G})$.
\end{restatable}
\begin{proof}[Idea]
  Converting a mixed $\ell$-expression for~$G$ into a directed $\ell$-expression for~$D(G)$ is straightforward, as we can replace each edge operation by two opposite arc operations. 
  The inverse requires a careful manipulation of the $\ell$-expression
  for~$D(G)$, to ensure that we can replace pairs of opposite arc
  operations by edge operations.

  The partial orientation~$\tilde{G}$ of the star~$K_{1,2}$, where the center vertex is connected to one leaf by an edge and to the other leaf by an outgoing arc, has~$\cwd(\tilde{G})=2$ and~$\cwm(\tilde{G})=3$.
\end{proof}

In order to generalize neighborhood diversity to mixed graphs without directed cycles, we first need to define the different types of neighbors that occur in a mixed graph. 
Given a vertex~$v$ of a mixed graph~$G$, the \emph{outgoing neighbors} of~$v$, denoted by~\emph{$\outneigh(v)$}, are all vertices~$u$ to which~$v$ has an \emph{outgoing arc}~$(v,u)$. Analogously, the \emph{incoming neighbors} of~$v$, denoted by~\emph{$\inneigh(v)$}, are all vertices~$u$ from which~$v$ has an \emph{incoming arc}~$(u,v)$. The \emph{undirected neighbors} of~$v$, denoted by~\emph{$\uneigh(v)$}, are all vertices~$u$ which are connected to~$v$ by an edge~$\set{u,v}$. 
Thus, the \emph{neighborhood} of a vertex~$v$, denoted by~\emph{$\neigh(v)$}, is the union of its incoming, outgoing, and undirected neighbors. 

The parameter neighborhood diversity is defined for undirected graphs via an equivalence relation on the vertices, where two vertices~$u,v$ are of the same type, denoted~$u\sim_\ndu v$, if they have the same neighbors, i.e., if~$\neigh(u)\setminus\set{v}=\neigh(v)\setminus\set{u}$. 
We generalize this equivalence relation to mixed graphs by further distinguishing between incoming, outgoing, and undirected neighbors. 
Therefore, two vertices~$u,v$ of a mixed graph~$G$ are of the same \emph{(mixed) type}, denoted~$u\sim_\ndm v$, if they have the same in-, out-, and undirected neighbors, i.e., if~$\uneigh(u)\setminus\set{v}=\uneigh(v)\setminus\set{u}$,~$\inneigh(u)=\inneigh(v)$, and~$\outneigh(u)=\outneigh(v)$.
\begin{restatable}[\restateref{lem:ndm-equiv-relation}]{lemma}{NDMEquivRelation}
  \label{lem:ndm-equiv-relation}
  For any mixed graph~$G$ (without directed cycles), the relation~$\sim_\ndm$ is an equivalence relation.
  Further, vertices of the same type induce an (undirected) clique or an independent set.
\end{restatable}

We call the partition of the vertices into equivalence classes
under~$\sim_\ndm$ a \emph{mixed neighborhood partition}. The
\emph{mixed neighborhood diversity} of a mixed graph~$G$, denoted
\emph{$\ndm(G)$}, is the number of equivalence classes
under~$\sim_\ndm$.
Note that this generalizes the definition of neighborhood diversity for undirected graphs. Furthermore, for mixed graphs without directed cycles, this is also a generalization of the definition of neighborhood diversity for directed graphs by Fernau et al.~\cite{FernauFMPN25ParaPathPartitions}, as in directed graphs without directed cycles there are no opposite arcs.
Note further that a mixed neighborhood partition can be obtained easily: it suffices to
check, for each pair of vertices, whether they are of the same type.

As is the case for mixed cliquewidth and cliquewidth, the mixed
neighborhood diversity can differ arbitrarily from the neighborhood
diversity of the underlying undirected graph since each vertex in a
directed path has to be of a different type, as we show now.

\begin{lemma}
  \label{lem:ndm-dir-path}
  Given a mixed graph~$G$ (without directed cycles) that contains a directed path of length~$\ell$, it holds that~$\ndm(G)\geq \ell+1$, with each vertex on the path being of a different type.
\end{lemma}
\begin{proof}
  We know from \cref{lem:ndm-equiv-relation} that there cannot be an arc between two vertices of the same type. 
  Let~$\angles{v_0,v_1,\dots,v_\ell}$ be a directed path of length~$\ell$ in~$G$.
  Suppose that there are two vertices~$v_i,v_j$ with~$i<j$ that are of the same type. 
  Since~$v_{j-1}\in \inneigh(v_j)$, it follows that~$v_{j-1}\in \inneigh(v_i)$. However, as~$i\leq j-1$, it follows that there exists a directed path from~$v_i$ to~$v_{j-1}$ in~$G$. Together with the arc~$(v_{j-1},v_i)$, this yields a directed cycle, a contradiction to~$G$ containing no directed cycles.
\end{proof}

\begin{proposition}\label{prop:ndm-ndu}
  For any mixed graph~$G$ (without directed cycles), $\ndu(G)\leq \ndm(G)$.
  Furthermore,~$\ndu$ does not bound~$\ndm$.
\end{proposition}
\begin{proof}
  It follows directly from the definitions that two vertices of the same type w.r.t.~$\ndm$ must also be of the same type w.r.t.~$\ndu$. 
  Indeed, if~$u\sim_\ndm v$, it holds that \[{\neigh(u)\setminus\set{v}=\inneigh(u)\cup\outneigh(u)\cup\uneigh(u)\setminus\set{v}=\inneigh(v)\cup\outneigh(v)\cup\uneigh(v)\setminus\set{u}=\neigh(v)\setminus\set{u}}\] (as there cannot be any arc between~$u$ and~$v$) and therefore~$u\sim_{\ndu}v$.
  
  The inverse does not hold, as two vertices of the same type w.r.t.~$\ndu$ may differ in their in- or out-neighbors, as the following family~$(G_\ell)_{\ell\geq 1}$ of graphs shows. The graph~$G_\ell$ is a partial orientation of the complete graph~$K_\ell$ obtained by orienting~$\ell-1$ edges such that the resulting mixed graph contains a directed Hamiltonian path. It holds that~$\ndu(G_\ell)=1$ but, due to \cref{lem:ndm-dir-path}, $\ndm(G_\ell)\geq\ell$.
\end{proof}

Lampis~\cite{Lampis12NeighborhoodDiversity} showed that neighborhood diversity is bounded by vertex cover, as the type of vertices not contained in the cover is determined uniquely by their neighbors (in the cover), resulting in a bound of~$\vc+2^{\vc}$. Since mixed neighborhood diversity further distinguishes between in-, out-, and undirected neighbors, we obtain the following bound.

\begin{proposition}
  \label{prop:ndm-vc}
  For any mixed graph~$G$ (without directed cycles), $\ndm(G)\leq \vc(G)+4^{\vc(G)}$.
\end{proposition}
\begin{proof}
  Let~$C$ be a vertex cover of~$G$ of size~$\vc(G)$. Since each vertex in~$V(G)\setminus C$ only contains neighbors in~$C$, there are at most~$4^{\vc(G)}$ different types of vertices in~$V(G)\setminus C$, as each vertex in~$C$ can be an in-, out-, or undirected neighbor, or not a neighbor at all. Further, there can be at most~$\vc(G)$ different types of vertices in~$C$, as~$C$ contains~$\vc(G)$ vertices. Thus, the total number of types is bounded by~$\vc(G)+4^{\vc(G)}$.
\end{proof}

If a parameter does not bound~$\ndu$ ($\cwu$) on undirected graphs (or vice versa), it also does not bound~$\ndm$ ($\cwm$) (or vice versa), as it holds that~$\ndm=\ndu$ ($\cwm=\cwu$) on undirected graphs.
Lampis~\cite{Lampis12NeighborhoodDiversity} showed that $\ndu$ bounds $\cw$ and that $\tw$ (and therefore $\cw$) does not bound $\ndu$ ($\ndu$ is even incomparable to $\tw$). It follows from this that $\cwm$ does not bound $\ndm$. Using a proof analogous to Lampis, we further obtain the following.
\begin{restatable}[\restateref{prop:ndm-cwm}]{proposition}{NDMtoCWM}
  \label{prop:ndm-cwm}
  For any mixed graph~$G$ (without directed cycles), $\ndm(G)+1\geq \cwm(G)$. 
  Further,~$\cwm$ does not bound~$\ndm$.
\end{restatable}

Regarding the relationship between $\cwm$ and $\ndu$, we already know that $\cwm$ does not bound $\ndu$ (due to $\cwu$ not bounding $\ndu$). However, while $\ndu$ bounds $\cwu$, it does not bound $\cwm$.
We show this by reusing the construction that we used to show that $\cwu$ does not bound $\cwm$ (\cref{prop:cwm-cwu}), as detailed in the following.

\begin{restatable}[\restateref{prop:cwm-ndu-incomp}]{proposition}{CWMtoNDUIncomp}
  \label{prop:cwm-ndu-incomp}
  The parameters $\cwm$ and $\ndu$ are incomparable.
\end{restatable}

\subparagraph{Transitive Closure}
Recall that an arc~$(u,v)$ is transitive if there exists a directed path from~$u$ to~$v$ using other arcs. The transitive closure~$G^+$ of a mixed graph~$G$ is the mixed graph obtained by adding every transitive arc to~$G$ (and removing any resulting parallel edges). 
As transitively closed graphs are a subset of mixed graphs, all bounds between parameters on mixed graphs also apply between their transitive closure counterparts. 
Further, if a parameter~$\alpha$ does not bound a parameter~$\beta$ on undirected graphs, then~$\alpha^+$ does not bound~$\beta^+$, as it holds that~$\alpha^+=\alpha$ and~$\beta^+=\beta$ on undirected graphs.
Thus, it remains to analyze the relationship between the parameters and their transitive closure counterparts.

All parameters considered in this paper, except for neighborhood diversity and cliquewidth, have the well-known property that they bound their value on any subgraph. Since the underlying graph of the transitive closure of a mixed graph is a subgraph of the underlying graph of the original graph, it follows that these parameters are bounded by their transitive closure counterparts. 

\begin{restatable}[\restateref{prop:param-tc-upper-bound}]{proposition}{ParamTCUpperBound}
  \label{prop:param-tc-upper-bound}
  For every mixed graph~$G$ (without directed cycles) and parameter $\alpha\in\set{\vc,\td,\fvs,\pw,\tw}$ it holds that $\alpha(G)\leq \alpha^+(G)$.
  Furthermore, these parameters do not bound their transitive closure counterparts.
\end{restatable}
\begin{proof}[Idea]
  The second part can be shown using a family of oriented stars, where half of the edges are oriented towards the center vertex and the other half are oriented away from the center vertex.
\end{proof}

However, for neighborhood diversity and cliquewidth (and their mixed variants), the situation is different. Contrary to the other parameters, these parameters are not bounded by their transitive closure counterparts, as shown in the following.

\begin{proposition}\label{prop:nd-cw-tc-not-upper-bound}
  The parameters~$\ndm$,~$\ndu$,~$\cwm$, and~$\cwu$ are not bounded by their transitive closure counterparts.
\end{proposition}
\begin{proof}
  We first construct a family~$(G_\ell)_{\ell\geq 1}$ of grid graphs where~$2\geq\cwm^+(G_\ell)\geq\cwu^+(G_\ell)$ but~$\cwm(G_\ell)\geq\cwu(G_\ell)=\ell+1$.
  The graph~$G_\ell$ is a partial orientation of the $(\ell\times \ell)$-grid graph where~$\ell^2-1$ edges are oriented such that~$G_\ell$ contains a directed Hamiltonian path, i.e., a directed path of length~$\ell^2-1$ containing all vertices. Golumbic and Rotics~\cite{GolumbicRotics00CliquewidthPerfectGraphClasses} showed that the cliquewidth of $(\ell\times \ell)$-grid graphs is~$\ell+1$. Together with \cref{prop:cwm-cwu} we obtain that~$\cwm(G_\ell)\geq\cwu(G_\ell)=\ell+1$. 
  As~$G_\ell$ contains a directed Hamiltonian path, it follows that~$G_\ell^+$ is the \emph{acylic tournament graph}, the directed graph with exactly one arc between each pair of vertices (and without directed cycles).
  Using the following claim, we obtain that~$2\geq\cwm^+(G_\ell)\geq \cwu^+(G_\ell)$.

  \begin{restatable}[\restateref{claim:cwm-dir-comp-graph}]{claim}{CWMDirCompGraph}
    The mixed cliquewidth of an acyclic tournament graph is (at most)~$2$.
    \label{claim:cwm-dir-comp-graph}
  \end{restatable}

  For neighborhood diversity, we construct a separate family~$(G'_\ell)_{\ell\geq 1}$ based on tripartite graphs. 
  The graph~$G'_\ell$ consists of three independent sets~$\set{u_1,\dots,u_\ell}$, $\set{v_1,\dots,v_\ell}$, and $\set{w_1,\dots,w_\ell}$. We add arcs~$(u_i,v_j)$ and~$(v_i,w_j)$ for~$i,j\in[\ell]$, such that every~$u$ has an outgoing arc to every~$v$ and every~$v$ has an outgoing arc to every~$w$. 
  To ensure that vertices of different sets are of different types, we further add three vertices~$u^*$,~$v^*$, and~$w^*$. We then add edges~$\set{u^*,u_i}$, $\set{v^*,v_i}$, and $\set{w^*,w_i}$ for all~$i\in[\ell]$. 
  Lastly, we add arcs~$(u_i,w_i)$ for~$i\in[\ell]$. 
  Thus, each vertex~$u_i$ has a distinct type, as does every vertex~$w_i$. Meanwhile, all vertices~$v_i$ are of the same type. Further, the vertices~$u^*$,~$v^*$, and~$w^*$ have distinct types.
  This results in~$\ndm(G'_\ell)=\ndu(G'_\ell)=2\ell+4$. However, in the transitive closure, each vertex~$u_i$ has an outgoing arc to every vertex~$w_j$ and therefore all~$u_i$ are of the same type and all~$w_i$ are of the same type, resulting in~$\ndm^+(G'_\ell)=\ndu^+(G'_\ell)=6$.
\end{proof}

In fact, both mixed neighborhood diversity and mixed cliquewidth bound their transitive closure counterparts, as shown in the following.

\begin{restatable}[\restateref{prop:ndm/cwm-upper-bound-tc}]{proposition}{NDMCWMUpperBoundTC}
  \label{prop:ndm/cwm-upper-bound-tc}
  For any mixed graph~$G$ (without directed cycles), it holds that $\ndm^+(G)\leq \ndm(G)$ and $\cwm^+(G)\leq 4^{\cwm(G)}\cdot\cwm(G)$.
\end{restatable}
\begin{proof}[Idea]
  For mixed neighborhood diversity, it holds that two vertices~$u,v$ of the same type in~$G$ are also of the same type in~$G^+$, as the set of vertices to (from) which~$u,v$ have outgoing (incoming) paths in~$G$ are equivalent.

  For mixed cliquewidth, we show that given a mixed $\ell$-expression constructing~$G$, we can obtain a $(4^\ell\cdot \ell)$-expression constructing~$G^+$. 
  The idea is to keep track of the labels of vertices to which a given vertex~$v$ has outgoing paths and the labels of vertices from which it has incoming paths. This can be achieved by using~$4^\ell$ additional labels for each original label. 
  These additional labels then allow us to add the transitive arcs, resulting in the claimed bound.
\end{proof}

The same does not hold for neighborhood diversity and cliquewidth, as these are in fact incomparable to their transitive closure counterparts, as the following shows.

\begin{proposition}\label{prop:ndu-cwu-tc-incomparable}
  The parameters~$\ndu$ and~$\cwu$ are incomparable to their transitive closure counterparts.
\end{proposition}
\begin{proof}
  We already know from \cref{prop:nd-cw-tc-not-upper-bound} that these parameters are not bounded by their transitive closure counterparts. Thus, it remains to show that they do not bound their transitive closure counterparts. For this, we construct a family of graphs~$(G_\ell)_{\ell\geq 2}$ such that~$\ndu(G_\ell)=\cwu(G_\ell)=2$ while~$\ndu^+(G_\ell)\geq\cwu^+(G_\ell)\geq\ell+1$. 

  The graph~$G_\ell$ consists of~$\ell^2$ \emph{grid vertices} and~$2\ell(\ell-1)$ \emph{arc vertices} and is constructed such that the grid vertices form an independent set in~$G_\ell$ but induce an $(\ell\times\ell)$-grid in~$G_\ell^+$. This is achieved through the usage of the arc vertices, which connect two grid vertices adjacent in~$G_\ell^+$ with a directed path, thus ensuring that the transitive arc is present in~$G_\ell^+$. Further, each arc vertex is connected to all grid vertices via edges, except for those to which it is connected via an arc. This ensures that all arc vertices are of the same type and that all grid vertices are of the same type, resulting in~$\ndu(G_\ell)=\cwu(G_\ell)=2$; see \cref{fig:grid-tc-example} for an example of the construction for~$\ell=3$.
  As shown by Golumbic and Roctis~\cite{GolumbicRotics00CliquewidthPerfectGraphClasses}, the cliquewidth of an $(\ell\times \ell)$-grid is~$\ell+1$. 
  Furthermore, it is easy to see that the cliquewidth of a graph bounds the cliquewidth of induced subgraphs.
  Thus, it follows that~$\ndu^+(G_\ell)\geq \cwu^+(G_\ell)\geq \ell+1$.
  \begin{figure}[tbp]
    \centering
    \includegraphics{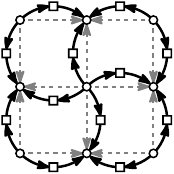}
    \caption{Example of the graph~$G_\ell$ for~$\ell=3$. The grid vertices are depicted as disks and the arc vertices as squares. The grid vertices form an independent set in~$G_\ell$ but induce a $(3\times 3)$-grid in~$G_\ell^+$ (the transitive arcs are depicted as dashed gray). The edges between arc vertices and grid vertices are omitted for visual clarity.}
    \label{fig:grid-tc-example}
  \end{figure}
\end{proof}

\subparagraph{Chromatic Number}
Clearly, any proper coloring of a mixed graph is also a proper coloring of the underlying undirected graph and thus the chromatic number of the underlying undirected graph lower bounds the chromatic number of the mixed graph.
Recall that the maxrank~$(\maxrank)$ of a mixed graph is the length of the longest directed path in the graph. 
For a given vertex~$v$, we define the \emph{inrank}, denoted as~\emph{$\inrank(v)$}, as the length of the longest directed path ending in~$v$. 
Furthermore, as vertices along a directed path must be colored increasingly, every vertex has a color that is at least its inrank. Thus, it holds that the chromatic number bounds the maxrank. 
We summarize these insights in the following proposition.

\begin{proposition}\label{prop:chi-lower-bounds}
  For any mixed graph~$G$ (without directed cycles), $\chi_\mathrm{u}(G)\leq \chi(G)$ and $\maxrank(G)\leq \chi(G)$.
\end{proposition}

For undirected graphs, several (upper) bounds on the chromatic number are known, such as treewidth and maximum degree. However, these bounds do not generalize to mixed graphs, as a simple directed path of length~$\ell$ has treewidth~$1$ and maximum degree~$2$, but chromatic number~$\ell+1$. 
Gutowski et al.~\cite[Theorem 9 \& Proposition 10]{GutowskiJKR0Z23ColorRecMixIntGraph} showed that for mixed interval graphs it holds that~$\chi\leq(\maxrank+1)\cdot\omega$, where~$\omega$ is the clique number, and that this bound is asymptotically tight. 
They achieved this by introducing the concept of \emph{layering}, which is a partition of the vertex set into~$\maxrank+1$ layers~$\angles{L_0,\dots,L_\maxrank}$ such that the layer~$L_i$ contains all vertices of inrank~$i$. In a layering, each arc is oriented towards the higher layer, as it must hold for every arc~$(u,v)$ that~$\inrank(u)<\inrank(v)$. It follows that each subgraph~$G[L_i]$ induced by a layer~$L_i$ is an undirected graph. Thus, we can iteratively color each layer with disjoint color sets, resulting in a proper coloring of the entire graph as formalized in the following theorem.

\begin{restatable}[\restateref{thm:chi-upper-bound}]{theorem}{ChiUpperBound}
  \label{thm:chi-upper-bound}
  For any mixed graph~$G$ (without directed cycles) and its corresponding layering~$\angles{L_0,L_1,\ldots,L_{\maxrank}}$, it holds that~$\chi(G)\leq\sum_{i=0}^{\maxrank}\chi_\mathrm{u}(G[L_i])$. Further, this bound is tight.
\end{restatable}

Regarding the transitive closure, note that~$\maxrank^+=\maxrank$ and~$\chi^+=\chi$, as adding transitive arcs neither changes the length of the longest directed path nor the validity of a coloring.
Thus, we obtain the following bounds on maxrank.

\begin{restatable}[\restateref{prop:rank-bounds}]{proposition}{RankBounds}
  \label{prop:rank-bounds}
  For any mixed graph~$G$ (without directed cycles), it holds that $\maxrank(G)+1\leq \ndm^+(G)$, $\maxrank(G)\leq \tw^+(G)$, $\maxrank(G)\leq 2^{\td(G)}-2$, and~$\maxrank(G)\leq 2\vc(G)$.
\end{restatable}

As the chromatic number of a subgraph is a lower bound for the chromatic number of the entire graph, it follows from \cref{thm:chi-upper-bound} and \cref{prop:rank-bounds} that any parameter that bounds the chromatic number of undirected graphs as well as the maxrank also bounds the chromatic number of mixed graphs.

\begin{restatable}[\restateref{cor:chi-bounds}]{corollary}{ChiBounds}
  \label{cor:chi-bounds}
  The chromatic number of mixed graphs is bounded by $\tw^+$ as well as $\td$ and $\vc$. 
  In fact, it even holds for every mixed graph~$G$ that~$\chi(G)\leq 2\vc(G)+1$.
\end{restatable}

The resulting hierarchy is visualized in \cref{fig:mix-para-hierarchy}.

\section{Parameterized Complexity of \MixedColoring{}}
\label{sec:para-complexity}

Recall that a problem is \emph{fixed-parameter tractable}
{\color{defblue}(\FPT)} with respect to a parameter~$k$ if it can be
solved in~$f(k)\cdot n^{\Oh(1)}$ time, where~$f$ is a computable
function and~$n$ is the size of the input.  A problem is
\emph{slice-wise polynomial} {\color{defblue}(\XP)} with respect
to~$k$ if it can be solved in~$\Oh(n^{f(k)})$ time (where $f$ and $n$
are as above).  Clearly, \FPT\ $\subseteq$ \XP.  Note that, given an
\FPT- or \XP-algorithm for $k$-(\textsc{Mixed})\Coloring\ that is
parameterized by the number of colors, $k$, we can easily obtain an
\FPT- or \XP-algorithm parameterized by the chromatic number.
Therefore, we state our results with respect to the parameter
chromatic number (if the runtime is \FPT or \XP parameterized by the
number of colors).

Courcelle's theorem yields for any problem that can be expressed in monadic second-order logic on graphs an \FPT{}-algorithm parameterized by tree- or even cliquewidth (depending on the logic variant used). 
Like cliquewidth, monadic second-order logic on graphs extends intuitively to mixed graphs, and expressing \MixedColoring{} as a monadic second-order formula is straightforward. However, to our knowledge, Courcelle's theorem has not yet been entirely proven for mixed graphs, with Arnborg et al.~\cite{ArnborgLS91MixedCourcelle} only proving the theorem for treewidth but not for cliquewidth. Still, by converting our mixed graphs into directed graphs as in \cref{prop:cwm-tw}, we can apply the directed version of Courcelle's theorem to obtain the following result.

\begin{restatable}[\restateref{thm:logic-fpt-cw-k}]{theorem}{LogicFPTCWK}
  \label{thm:logic-fpt-cw-k}
  The problem \MixedColoring{} is \FPT{} parameterized by mixed cliquewidth plus chromatic number.
\end{restatable}

While Courcelle's theorem yields an \FPT{}-algorithm for \MixedColoring{}, the resulting algorithm is quite impractical. For example, the runtime contains huge constants depending on the length of the formula, i.e., the number of colors, which cannot be bounded by an elementary function unless~$\Poly=\NP$, as shown by Frick and Grohe~\cite{FrickGrohe04CourcelleConstants}. Since polynomials are elementary functions, this also implies that this algorithm is not an \XP{}-algorithm parameterized solely by tree- or cliquewidth.

Another \Coloring{} algorithm can also be generalized to \MixedColoring{}: it
is well known~\cite[Theorem
7.9]{CyganFKLMPPS15ParameterizedAlgorithms} that \Coloring, given a
tree-decomposition of width~$\tw$, can be solved in~$\Oh^*(k^\tw)$
time via a dynamic program. By exchanging the checks for a proper
coloring with checks for a proper mixed coloring (i.e., additionally
checking that no arc is violated), \MixedColoring{} can be solved
within the same time bound, as we now summarize.

\begin{theorem}
  \label{thm:tw-algo-coloring}
  Given a mixed graph and a tree-decomposition of width~$\tw$, the
  problem $k$-\MixedColoring\ can be solved in~$\Oh^*(k^\tw)$ time.
\end{theorem}

As an optimal tree-decomposition can be computed in \FPT-time w.r.t.\ treewidth~\cite{Bodlaender96FPTTW} and the chromatic number is bounded by the number of vertices, this yields \FPT and \XP algorithms for \MixedColoring{}, as summarized in the following.

\begin{corollary}
  \label{cor:tw-fpt-xp}
  \MixedColoring{} is \FPT{} parameterized by treewidth plus chromatic number and \XP{} parameterized solely by treewidth.
\end{corollary}

Note that this algorithm is an improvement over the \XP-algorithm developed by Ries and de~Werra~\cite{Ries-deWerra08TwoColorProbMixGraph}, which runs in~$\Oh(n^{2\tw+4}m^{\tw+2})$ time, where~$m$ is the number of edges and~$n$ the number of vertices.

We cannot hope for an \FPT{}-algorithm parameterized by treewidth alone, as we show in the following that \MixedColoring{} is \W{1}-hard w.r.t.\ pathwidth and feedback vertex set, which are parameters that bound treewidth. We achieve this via a parameterized reduction from \ListColoring{}, a generalization of \Coloring{} where each vertex can only be assigned a color from a given list. While Fellows et al.~\cite{FellowsFLRSST11ComplexColorParaTW} only showed that \ListColoring{} is \W{1}-hard w.r.t.\ treewidth, their reduction actually even yields \W{1}-hardness w.r.t.\ vertex cover, as we show in the following lemma.

\begin{restatable}[\restateref{thm:list-coloring-w1-vc}]{lemma}{ListColoringWOneVC}
  \label{thm:list-coloring-w1-vc}
  \ListColoring{} parameterized by vertex cover is \W{1}-hard.
\end{restatable}

To reduce \ListColoring{} to $k$-\MixedColoring, where~$k$ is the
total number of distinct colors in the \ListColoring{} instance, we enforce the color
list of each vertex by adding adequate edges to newly inserted
directed paths of vertex-length~$k$. This yields the following.

\begin{restatable}[\restateref{lem:para-list-coloring-to-mixed-coloring}]{lemma}{ParaListColoringToMixedColoring}
  \label{lem:para-list-coloring-to-mixed-coloring}
  There exists a parameterized reduction w.r.t.\ pathwidth as well as w.r.t.\ feedback vertex set from \ListColoring{} to \MixedColoring{}.
\end{restatable}

Combining the \W{1}-hardness of \ListColoring{} with the parameterized reduction to \MixedColoring{} we obtain the following theorem.

\begin{theorem}
    \label{thm:w1-fvs-pw}
    \MixedColoring{} is \W{1}-hard w.r.t.\ feedback vertex set as well as w.r.t.\ pathwidth.
\end{theorem}

Recall that \Coloring{} is \FPT{} with respect to neighborhood
diversity~\cite{Ganian12NeighborDiversityColoringFPT} and \XP{} with
respect to cliquewidth~\cite{KoblerRotics03XPColoringCliqueWidth}.
In contrast to that, we now show that \MixedColoring{} is \paraNP-hard
with respect to both of these parameters, even in more restricted
settings.

\begin{restatable}[\restateref{thm:para-np-cwm-rank}]{theorem}{MixedColoringParaNPCW}
  \label{thm:para-np-cwm-rank}
  \MixedColoring{} is \paraNP-hard w.r.t.\ mixed cliquewidth plus maxrank.
\end{restatable}
\begin{proof}[Idea]
  Given a set of strings~$\mathcal{S}$ over the binary
  alphabet~$\set{0,1}$ and a positive integer~$k$, the problem
  \PShortestSuperstring asks whether there is a string of length at
  most~$k$ that is a superstring of each string
  in~$\mathcal{S}$.  \PShortestSuperstring is
  \NP-hard~\cite{RaihaUkkonen81SCSBinaryNPComplete}.

  We reduce from \PShortestSuperstring{} by constructing a mixed
  graph~$G$ such that a proper $k$-coloring of~$G$ corresponds to a
  length-$k$ superstring of~$\mathcal{S}$.  We achieve this by
  creating a directed path~$P_S$ for each string~$S$ in~$\mathcal{S}$
  such that the vertices of the path correspond to the characters of
  the string; we add edges such that vertices corresponding to
  different characters cannot receive the same color.  See
  \cref{fig:sss-reduction-instance} for an example.  Thus, the color
  of a vertex corresponds to the position of the corresponding
  character in the superstring, and the directions of the paths
  enforce that the characters of each string appear in the correct
  order.
  
  \begin{figure}[tbp]
    \centering
    \captionsetup[subfigure]{justification=centering}
    \begin{subfigure}[b]{.45\textwidth}
      \centering
      \includegraphics{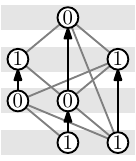}
      \caption{}
      \label{fig:sss-reduction-instance}
    \end{subfigure}
    \begin{subfigure}[b]{.45\textwidth}
      \centering
      \includegraphics{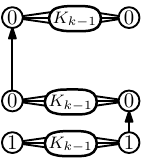}
      \caption{}
      \label{fig:sss-ladder}
    \end{subfigure}
    \caption{
      (a) The mixed graph resulting from an instance of \PShortestSuperstring with strings~$\set{01,100,11}$ and shortest superstring~$1010$.
      (b) The construction resulting from splitting the path corresponding to the string~$100$ at each vertex and connecting the resulting vertices to distinct cliques to enforce that they receive the same color in a proper $k$-coloring.
    } 
    \label{fig:sss-reduction}
  \end{figure}

  To ensure a constant maxrank, we then split each directed path~$P_S$ at every vertex into two vertices, with one vertex obtaining the incoming and the other the outgoing arc. We then connect each new pair of vertices to a distinct $(k-1)$-clique, thus enforcing that the two vertices receive the same color in any proper $k$-coloring; see \cref{fig:sss-ladder}. Therefore, regarding proper $k$-colorings, they can be treated as a single vertex, while the maxrank of the resulting graph is~$1$, due to all directed paths having been split.
  Furthermore, the resulting graph has a mixed cliquewidth of at most~$6$.
\end{proof}

\begin{theorem}
  \label{thm:paraNP-ndu}
  \MixedColoring{} is \paraNP-hard parameterized by neighborhood
  diversity, even if considering the neighborhood diversity of the
  transitive closure.
\end{theorem}
\begin{proof}
  We reduce from the following problem, which is
  \NP-hard~\cite[SS9]{GareyJohnson79ComputersIntractabilityGuideTheoryNPCompleteness}.

  \defdecproblem{\PPrecedenceConstrainedScheduling}
  {Sets~$T_1,T_2$ of unit-length tasks, a partial order~$\prec$
    on~$T_1\cup T_2$, an integer~$D$.}
  {Can the tasks be scheduled with makespan~$D$ on two parallel
    machines~$M_1$ and~$M_2$ such that the tasks in~$T_1$ are assigned
    to~$M_1$, the tasks in~$T_2$ are assigned to~$M_2$, and all
    precedence constraints are respected (i.e., if $t\prec t'$, then
    $t$ must finish before~$t'$ starts)?}

  Given an instance $(T_1,T_2,\prec,D)$ of
  \PPrecedenceConstrainedScheduling, we create a $4D$-\MixedColoring{}
  instance~$G$ such that each block of four consecutive colors corresponds
  to one time unit on the machines.

  First, we create a directed path~$P$ consisting of~$4D$ vertices~$\angles{p_1,\dots,p_{4D}}$, which we will use to ensure that each task is scheduled on its assigned machine. Note that in any proper $4D$-coloring, the path~$P$ has the same unique coloring, with each vertex~$p_i$ receiving the color~$i$.
  For each task~$t\in T_1 \cup T_2$, we create two \emph{task
    vertices}, the \emph{start vertex}~$v_t^-$ and the \emph{end
    vertex}~$v_t^+$, connected via the arc~$(v_t^-,v_t^+)$.
  We denote with~$[i]_k$ the set of integers that are congruent to~$i$ modulo~$k$.
  For each start vertex~$v_t^-$ of a task~$t$ in~$T_1$, we add an edge
  to each vertex~$p_i$ of the path~$P$ with~$i\notin [1]_4$. As a
  result, $v_t^-$ receives a color in~$[1]_4$ in every proper
  $4D$-coloring of~$G$. Analogously, we add edges between each end
  vertex~$v_t^+$ of a task~$t$ in~$T_1$ and each vertex~$p_i$ of the
  path~$P$ with~$i\notin [3]_4$ to ensure that~$v_t^+$ receives a
  color in~$[3]_4$. For tasks in~$T_2$, we similarly add edges to
  ensure that the start vertices receive colors in~$[2]_4$ and the end
  vertices receive colors in~$[4]_4$.
  The precedence constraints are enforced by adding arcs between the corresponding tasks, i.e., for each~$t,t'\in T_1\cup T_2$ with~$t\prec t'$ we add the arc~$(v_t^+,v_{t'}^-)$.
  Lastly, it remains to ensure that no two tasks are scheduled at the same time on the same machine. We achieve this by adding (undirected) edges such that the underlying undirected graph induced by the task vertices forms a clique. 
  An example of the resulting graph is shown in \cref{fig:pcs-reduction}.

  \begin{figure}[tbp]
    \centering
    \includegraphics{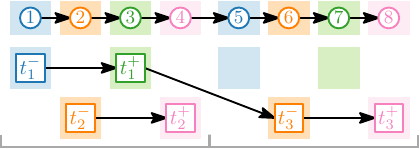}
    \caption{The mixed graph resulting from an instance of \PPrecedenceConstrainedScheduling with~$T_1=\set{t_1}$ and~$T_2=\set{t_2,t_3}$ as well as precedence constraints~$t_1\prec t_3$ and deadline~$2$. Four consecutive colors correspond to one time unit on the machines, and the colored slots correspond to the colors that the task vertices of the corresponding type can receive. 
    The edges are omitted for visual clarity.}
    \label{fig:pcs-reduction}
  \end{figure}

  \begin{restatable}[\restateref{claim:correctness-pcs}]{claim}{CorrectnessPCS}
    \label{claim:correctness-pcs}
    The resulting graph~$G$ has a proper $4D$-coloring if and only if there is a schedule of the tasks on their assigned machines that respects the precedence constraints and finishes by time~$D$.
  \end{restatable}

  \begin{restatable}[\restateref{claim:nd-tc-pcs}]{claim}{NDTCPCS}
    \label{claim:nd-tc-pcs}
    The transitive closure of~$G$ has a neighborhood diversity of~$8$.
  \end{restatable}

  As these claims show, the resulting \MixedColoring{} instance is equivalent to the original \PPrecedenceConstrainedScheduling{} instance, with the neighborhood diversity of the transitive closure being~$8$. Thus, \MixedColoring{} is \paraNP-hard{} w.r.t.\ neighborhood diversity, even when considering the neighborhood diversity of the transitive closure.
\end{proof}

The picture changes if we parameterize by {\em mixed} neighborhood diversity.

\begin{theorem}
  \label{thm:fpt-ndm}
  \MixedColoring{} is \FPT{} parameterized by mixed neighborhood diversity.
\end{theorem}
\begin{proof}
  Let~$G$ be a mixed graph. We can assume w.l.o.g.\ that each type
  of~$G$ induces a clique, as all vertices of a type inducing an
  independent set can be merged into a single vertex.  This is due to
  the fact that they have the same neighborhoods and therefore can be
  assigned the same color.
  Let~$\mathcal{S}=\set{C_1,\ldots,C_{\ndm}}$ be the set of types
  of~$G$, where each~$C_i$ is a clique in~$G$.  In the following, we
  write~$\set{C_i,C_j}\in E(G)$ and~$(C_i,C_j)\in A(G)$ to express
  that there is an undirected edge or a directed arc, respectively,
  between the types~$C_i$ and~$C_j$ in~$G$.
  For a coloring~$c$ of~$G$ and a set $U \subseteq V(G)$, let
  $c(U)=\{c(u) : u \in U\}$.  If~$c$ is proper, then, for
  every~$(C_i,C_j)\in A(G)$, all colors in~$c(C_i)$ are smaller than
  all colors in~$c(C_j)$.
  We define a \emph{type-endpoint preorder}~$p$ as a function that
  assigns, to each type~$C$, two integers~$p^-(C),p^+(C)\in[\ell]$
  (with~$p^-(C)<p^+(C)$ and $\ell\le 2\,\ndm$), which we call the
  \emph{type endpoints} of~$C$.
  We say that a coloring~$c$ \emph{corresponds} to a type-endpoint
  preorder~$p$ if there is a sequence $\angles{c_1,\dots,c_\ell}$
  of~$\ell$ ascending (not necessarily consecutive) colors such that,
  for each type~$C$, the set~$c(C)$ is contained in the halfopen
  interval~$[c_{p^-(C)}, c_{p^+(C)})$.
  We say that a (type-endpoint) preorder~$p$ is \emph{proper} if, for each arc~$(C_i,C_j)\in A(G)$, it holds that~$p^+(C_i)\leq p^-(C_j)$. An example of a proper preorder and a corresponding coloring is shown in \cref{fig:nd-ilp}.

  \begin{restatable}[\restateref{claim:proper-preorder-colorings}]{claim}{ProperPreorderColorings}
    \label{claim:proper-preorder-colorings}
    Given a proper (type-endpoint) preorder~$p$, any coloring that corresponds to~$p$ does not violate any arcs.
    Further, any proper coloring corresponds to at least one proper preorder.
  \end{restatable}

  \begin{figure}[tbp]
    \centering
    \includegraphics{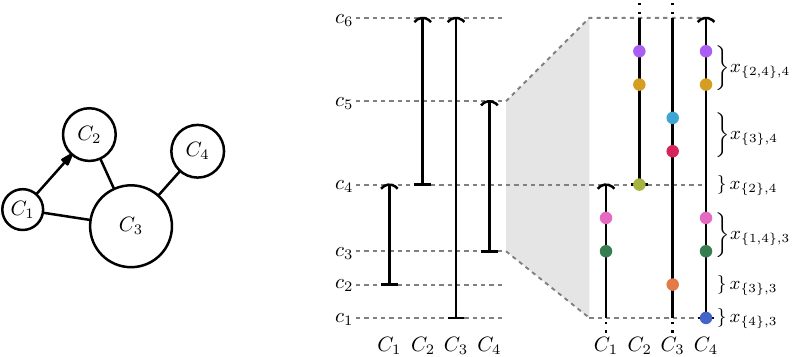}
    \caption{A mixed graph with four types~$C_1,\dots,C_4$ and a
      preorder with~$\ell=6$ type endpoints. Since it holds for its
      only arc~$(C_1,C_2)$ that~$p^+(C_1)=c_4=p^-(C_2)$, the preorder is
      proper.  On the right, we depict a part of a coloring
      corresponding to the preorder, together with the variables of
      the \ILP{}.
      Note that most of the variables are zero, either due to a constraint of the \ILP{}, such as for~$x_{\set{1},4}$ or~$x_{\set{3,4},3}$, or simply due to the fact that there are no colors shared exclusively by the given set of types in the given interval, such as for~$x_{\set{1},3}$.
      }
    \label{fig:nd-ilp}
  \end{figure}

  In total, there are~$2\,\ndm$ type endpoints and thus at most~$(2\,\ndm)!$ orders of these endpoints. Since type endpoints can be equal, there are at most~$(2\,\ndm)!\cdot 2^{2\,\ndm}\in 2^{\Oh(\ndm\log\ndm)}$ (proper) preorders.
  Therefore, we can enumerate all proper preorders in \FPT-time
  w.r.t.\ mixed neighborhood diversity.  We can thus check whether~$G$
  is $k$-colorable by checking, for each proper preorder~$p$, whether
  there exists a proper $k$-coloring of~$G$ that corresponds to~$p$.
  Since we know that such a coloring does not violate any arcs, it
  suffices to check whether it does not violate any edges.  This
  problem has similarities to \Coloring{}, differing in that we
  additionally require the colors within each type to adhere to the
  given preorder.  For \Coloring{},
  Ganian~\cite{Ganian12NeighborDiversityColoringFPT} provided an
  \FPT-algorithm w.r.t.\ neighborhood diversity based on an \ILP. In
  the following, we adapt his \ILP{} to additionally respect a given
  proper preorder~$p$.
  
  Since we only need to decide whether there exists a $k$-coloring~$c$
  of~$G$ that corresponds to the given proper preorder~$p$, we do not
  need an objective function. For an overview of the ILP,
  see the \arxiv{proof of \cref{claim:ilp-correctness} in the appendix}{full version~\cite{FullVersion}}. 
  We have two kinds of variables. 
  First, for every~$i\in[\ell]$, we introduce an integer variable $c_i\in[k]$ 
  to obtain a sequence~$\angles{c_1,\dots,c_\ell}$ of colors.
  We enforce that these colors form an ascending sequence (by adding, for each~$i\in[\ell-1]$, the constraint~$c_i+1 \le c_{i+1}$). 
  Second, for each~$i\in[\ell-1]$ and each subset~$\mathcal{S}'\subseteq \mathcal{S}$, we introduce an integer variable~$x_{\mathcal{S}',i}\in\set{0,\dots,k}$, which represents the number of colors in the interval~$[c_i,c_{i+1})$ that are shared exclusively by the types in~$\mathcal{S}'$.
  Formally, we want that ${x_{\mathcal{S'},i}=\abs{\set{d\in[c_i,c_{i+1})\mid (\forall C\in
    \mathcal{S}'\colon d\in c(C))\land(\forall
    C\in\mathcal{S}\setminus\mathcal{S}'\colon d\notin
    c(C))}}}$, where~$c(C)$ is the set of colors assigned to the vertices of type~$C$ by the coloring~$c$; see \cref{fig:nd-ilp} for an example. 
  For each~$i\in[\ell-1]$, we need to ensure that the ILP does not assign too many colors in the interval~$[c_i,c_{i+1})$. We achieve this by adding the constraint ${\sum_{\mathcal{S}'\subseteq\mathcal{S}}x_{\mathcal{S}',i} \le c_{i+1}-c_i}$.
  Moreover, we want the coloring to correspond to the given
  preorder~$p$.  To achieve this, for each type~$C$, we first ensure
  that no color~$d < c_{p^-(C)}$ is assigned to~$C$ (by adding the
  constraint~$\sum_{i=1}^{p^-(C)-1}\sum_{\mathcal{S}'\subseteq
    \mathcal{S}\colon C\in \mathcal{S}'}x_{\mathcal{S}',i}=0$).
  Second, we ensure that enough colors are assigned to~$C$ in the
  interval formed by its two type endpoints (by adding the
  constraint~$\sum_{i=p^-(C)}^{p^+(C)-1}\sum_{\mathcal{S}'\subseteq
    \mathcal{S}\colon C\in \mathcal{S}'}x_{\mathcal{S}',i}=\abs{C}$).
  Third, we ensure that no color $d \ge c_{p^+(C)}$ is assigned to~$C$ (by
  adding the
  constraint~$\sum_{i=p^+(C)}^{\ell-1}\sum_{\mathcal{S}'\subseteq
    \mathcal{S}\colon C\in \mathcal{S}'}x_{\mathcal{S}',i}=0$).
  It remains to ensure that no two types connected by an edge share a
  color.  To this end, for each~$i\in[\ell-1]$ and each
  set~$\mathcal{S}'\subseteq \mathcal{S}$ that contains types~$C$
  and~$C'$ with $\{C,C'\} \in E(G)$, we add the
  constraint~$x_{\mathcal{S}',i}=0$.

  \begin{restatable}[\restateref{claim:ilp-correctness}]{claim}{ILPCorrectness}
    \label{claim:ilp-correctness}
    For a given proper preorder~$p$, there is a feasible solution to the constructed \ILP{} if and only if there is a proper $k$-coloring of~$G$ that corresponds to~$p$.
  \end{restatable}
  
  It is well known~\cite[Theorem
  9.19]{FlumGrohe06ParameterizedComplexityTheory} that an \ILP{} can
  be solved in \FPT-time w.r.t.\ the number of variables.  Hence, we
  can solve the \ILP{} in \FPT-time w.r.t.\ mixed neighborhood
  diversity as it uses $\Oh(2^{\ndm} \cdot \ndm)$ variables.
  Therefore, by enumerating all proper preorders and checking for each
  preorder whether a corresponding proper $k$-coloring exists, we can
  decide in \FPT-time w.r.t.\ mixed neighborhood diversity whether~$G$
  is $k$-colorable.
\end{proof}

Furthermore, while \MixedColoring{} is \paraNP-hard w.r.t.\ neighborhood diversity, it becomes \FPT{} if we use the chromatic number as an additional parameter.

\begin{theorem}
  \label{thm:fpt-ndu-chrom}
  \MixedColoring{} is \FPT{} parameterized by neighborhood diversity plus chromatic number.
\end{theorem}
\begin{proof}
  \begin{figure}[tbp]
    \centering
    \includegraphics{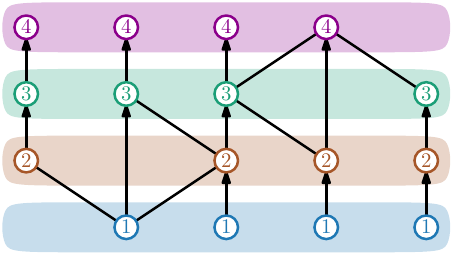}
    \caption{A uniquely $4$-colorable mixed graph where no color class is a maximal independent set.}
    \label{fig:mis-counterexample}
  \end{figure}

  It was shown by Christofides~\cite{Christofides71AlgoChromNum} that there exists an optimal coloring of an (undirected) graph where one color class forms a maximal independent set. However, this property does not hold for mixed graphs; see \cref{fig:mis-counterexample}. Fortunately, we can strengthen the property by only considering maximal independent sets of vertices with inrank~$0$, i.e., vertices without incoming arcs, resulting in the following formula for the chromatic number.

  \begin{restatable}[\restateref{claim:mis0-coloring-property}]{claim}{MISZColoringProp}
    \label{claim:mis0-coloring-property}
    For every (non-empty) mixed graph~$G$ there exists an optimal coloring of~$G$ where the first color class forms a maximal independent set in the subgraph induced by the vertices of inrank~$0$. 
    Thus, for every mixed graph~$G$ it holds that~$\chi(G)=0$ if~$G$ is empty and~$\chi(G)=\min_{I}\chi(G-I)+1$ otherwise, where the minimum is taken over all maximal independent sets~$I$ in the subgraph induced by the vertices of inrank~$0$.
  \end{restatable}

  Furthermore, we can bound the number of maximal independent sets in the subgraph induced by vertices of inrank~$0$ by neighborhood diversity and clique number, as the following shows.

  \begin{restatable}[\restateref{claim:mis0-bound}]{claim}{MISZBound}
    \label{claim:mis0-bound}
    For every (non-empty) mixed graph~$G$ it holds that the number of maximal independent sets in the subgraph induced by the vertices of inrank~$0$ is at most~$(\omega(G)+1)^{\ndu(G)}$.
  \end{restatable}

  We implement the recursive formula for the chromatic number as a branching algorithm for \MixedColoring{} as follows: Given a mixed graph~$G$ and a number of colors~$k$, we branch on all maximal independent sets~$I$ in the subgraph induced by the vertices of inrank~$0$ and recursively check whether~$G-I$ is $(k-1)$-colorable until we either reach the empty graph or run out of colors. Thus, the recursion depth is bounded by~$k$. As in each step we branch into at most~$(\omega(G)+1)^{\ndu(G)}$ subproblems, this results in an overall runtime of~$\Oh^*((\omega(G)+1)^{\ndu(G)\cdot k})$, which is \FPT{} w.r.t.\ neighborhood diversity plus chromatic number, as~$\chi(G)\geq \omega(G)$.
\end{proof}

\section{Open Problems}
\label{sec:conclusion}
While we defined mixed neighborhood diversity and mixed cliquewidth only for mixed graphs without directed cycles, these parameters can also be defined for mixed graphs with directed cycles. For mixed cliquewidth the definition remains the same, while for mixed neighborhood diversity the definition can be adapted by integrating the restrictions for opposite arcs introduced by Fernau et al.~\cite{FernauFMPN25ParaPathPartitions}. The relationships of these more general parameters should be analogous; most proofs can be generalized directly since they do not depend on the requirement that mixed graphs do not contain directed cycles.
We remark that it seems likely that Courcelle's theorem extends to
(general) mixed cliquewidth, which would facilitate its application to
other problems on mixed graphs.  
Further, we think that it would be interesting
to investigate how graph parameters and problem complexities change if
we remove all transitive arcs from input graphs.  
Preliminary computational experiments indicate that \MixedColoring{} algorithms benefit from such a data reduction step.
Lastly, while we explored the parameterized complexity of \MixedColoring{},
exact (non-parameterized) algorithms for \MixedColoring{} remain
largely unexplored (except for ILP-based approaches).

\bibliography{paper}

\clearpage
\appendix

\section{Parameters}
\label{appx:parameters}

The following parameters are defined for undirected graphs, and applied to mixed graphs by taking the parameter value of the underlying undirected graph.

\begin{definition}[Vertex Cover]
  A \emph{vertex cover} of a graph~$G$ is a set of vertices~$C$, such that every edge~$\set{u,v}\in E(G)$ has at least one endpoint in~$C$. The \emph{vertex cover (number)} of~$G$, denoted by~\emph{\vc(G)}, is the size of the smallest vertex cover of~$G$.
\end{definition}

\begin{definition}[Treedepth]
  The \emph{treedepth} of a graph~$G$, denoted by~$\emph{\td(G)}$, is the minimum height of a rooted forest such that for every edge~$\set{u,v}\in E(G)$, the vertices~$u$ and~$v$ have an ancestor-descendant relationship in the forest. The \emph{height} of a rooted forest is the maximum number of vertices on a path from a leaf to the root.
\end{definition}

\begin{definition}[Feedback Vertex Set]
  A \emph{feedback vertex set} of a graph~$G$ is a set of vertices~$F$, such that~$G[V(G)-F]$ contains no cycle. The \emph{feedback vertex set (number)}, denoted by~$\emph{\fvs(G)}$, is the size of the smallest feedback vertex set of~$G$.
\end{definition}

\begin{definition}[Treewidth]
  A \emph{tree-decomposition} of a graph~$G$ is a pair~$\mathcal{T}=(T,\mathcal{B})$ with $\mathcal{B}=\set{B_t}_{t\in{V(T)}}$, where~$T$ is a tree whose every node~$t$ is assigned a subset of vertices~$B_t$, called \emph{bag}, such that the following conditions hold:
  \begin{itemize}
      \item $\bigcup_{t\in V(T)}B_t=V(G)$, i.e., each vertex is in at least one bag.
      \item For every edge~$\set{u,v}\in E(G)$ there is a node~$t\in V(T)$ such that~$u,v\in B_t$.
      \item For every vertex~$v\in V(G)$, the set of nodes containing~$v$,~$\set{t\in V(T)\mid v\in B_t}$, induces a connected subtree of~$T$.
  \end{itemize}
  The \emph{width of a tree-decomposition}~$\mathcal{T}=(T,\mathcal{B})$ is~$\max_{t\in V(T)}\abs{B_t}-1$, i.e., the maximum size of its bags minus~$1$. The \emph{treewidth} of a graph~$G$, denoted by~$\emph{\tw(G)}$, is the minimum possible width of a tree-decomposition of~$G$.
\end{definition}

\begin{definition}[Pathwidth]
  A \emph{path-decomposition} is a special case of a tree-decomposition $\mathcal{T}=(T,\mathcal{B})$, where we restrict~$T$ to be a path. The \emph{pathwidth} of a graph~$G$, denoted by~$\emph{\pw(G)}$, is the minimum possible width of a path-decomposition of~$G$.
\end{definition}

The following parameters are defined for mixed graphs without directed cycles. As they are generalizations of the parameters for undirected graphs, we only provide the definitions for mixed graphs.

\begin{definition}[Mixed Cliquewidth]
  The \emph{mixed cliquewidth} of a mixed graph~$G$, denoted by~$\emph{\cw(G)}$, is the minimum number of distinct labels needed to construct~$G$ using the following operations:
  \begin{description}
    \item[\emph{Introduce:}] creation of a new vertex with label~$i$, denoted~$\emph{i}$ 
    \item[\emph{Union:}] disjoint union of two graphs, denoted~\emph{$\cwunion$}
    \item[\emph{Relabel:}] renaming all labels~$i$ to~$j$, denoted by~\emph{$\cwrelabel_{i\to j}$}
    \item[\emph{Edge:}] adding an (undirected) edge between every vertex with label~$i$ and every vertex with label~$j$, for~$i\neq j$, denoted~\emph{$\cwedge_{i,j}$}
    \item[\emph{Arc:}] adding a (directed) arc from every vertex with label~$i$ to every vertex with label~$j$, for~$i\neq j$, denoted~\emph{$\cwarc_{i,j}$}
  \end{description}
\end{definition}

\begin{definition}[Mixed Neighborhood Diversity]
  Given a mixed graph~$G$, we say that two vertices~$u,v$ of~$G$ are of the same \emph{type}~$\sim_\ndm$ if and only if it holds that ${\inneigh(u)=\inneigh(v)}$, $\outneigh(u)=\outneigh(v)$, and~$\neigh(u)\setminus\set{v}=\neigh(v)\setminus\set{u}$. 
  The equivalence classes \emph{(types)} under~$\sim_\ndm$ form the \emph{mixed neighborhood partition} of~$G$. The \emph{mixed neighborhood diversity} of~$G$, denoted by~$\emph{\ndm(G)}$, is the number of equivalence classes under~$\sim_\ndm$.
\end{definition}

\newpage
\section{Full Proofs}

\CWMtoTW*
\label{prop:cwm-tw*}
\begin{proof}
  As shown by Courcelle and Olariu~\cite{CourcelleOlariu00BoundsCliqueWidth}, it holds for every directed graph~$G$ that $\cwd(G)\leq 2^{\tw(G)+1}+1$. Therefore, it suffices to show the relations between mixed cliquewidth and directed cliquewidth.

  Let~$G$ be a mixed graph with a corresponding mixed $\ell$-expression constructing~$G$. We can replace each edge operation~$\eta_{i,j}$ by two arc operations~$\cwarc_{i,j}$ and~$\cwarc_{j,i}$, thus obtaining a directed $\ell$-expression constructing~$D(G)$. Thus, it holds that~$\cwd(G)\leq \cwm(G)$.

  To show that~$\cwm(G)\leq 2\cdot \cwd(G)$, we show how to transform a directed $\ell$-expression constructing~$D(G)$ into a mixed $2\ell$-expression constructing~$G$. 
  First, we have to transform the given expression into a suitable form, which allows us to easily characterize which arc operations have to be replaced by edge operations. 
  We call a directed expression \emph{irredundant}, if for any arc operation~$\cwarc_{i,j}$ there are no arcs from~$i$ to~$j$ already present in the graph at the time this operation is performed. 
  We call an irredundant expression \emph{canonical}, if between any two union operations the arc operations are performed before the relabel operations, i.e., the expressions are of the form $\cwunion(\cwrelabel\dots\cwrelabel(\cwarc\dots\cwarc(\ldots\cwunion\dots)))$.
  \begin{claim}
    Given a directed $\ell$-expression constructing a directed graph~$\vec{G}$, there exists a canonical irredundant directed $\ell$-expression constructing~$\vec{G}$.
  \end{claim}
  \begin{claimproof}
    Courcelle and Olariu~\cite{CourcelleOlariu00BoundsCliqueWidth} showed by induction how to obtain an irredundant expression of the same width from any given expression. 
    Further, they showed that we can push arc operations downwards past relabel operations, i.e., we can transform an expression with a sequence~$\cwarc(\cwrelabel(\dots))$ into an equivalent expression with sequence~$\cwrelabel(\cwarc(\dots))$. Applying this to an irredundant expression, we obtain an equivalent canonical irredundant expression.
  \end{claimproof}
  We say that an expression is \emph{separated} if for every union operation~$G_1\oplus G_2$ it holds that the labels of~$G_1$ and~$G_2$ are disjoint at the time of the union operation. We say that a separated expression is \emph{canonical} if it holds for every union operation, that for each of the two subgraphs, no more arcs added between vertices of each subgraph after the union operation, i.e., all arcs contained in each induced subgraph are added before the union operation.
  \begin{claim}
    Given a canonical irredundant directed $\ell$-expression constructing a directed graph~$\vec{G}$, there exists a canonical separated irredundant $(2\ell)$-expression constructing~$\vec{G}$.
  \end{claim}
  \begin{claimproof}
    Courcelle and Olariu~\cite{CourcelleOlariu00BoundsCliqueWidth} showed by induction how to obtain a separated expression using at most~$2\ell$ labels from any given expression using~$\ell$ labels. Using the previous claim, we can transform this expression into a canonical irredundant expression while maintaining separability. Since our expression is separated, it holds for each arc operation w.r.t.\ the previous union operation~$G_1\oplus G_2$, that it either adds arcs only in~$G_1$,~$G_2$, or between~$G_1$ and~$G_2$. Thus, we can push all arc operations that work exclusively on~$G_1$ and~$G_2$ downwards past the union operation and the subsequent relabel operations, yielding a canonical separated irredundant expression.
  \end{claimproof}

  Therefore, we can assume w.l.o.g.\ that we are given a canonical separated irredundant directed $2\ell$-expression constructing the directed graph corresponding to~$G$. Note that this expression already adds all the arcs of~$G$. Thus, it remains to deal with edges of~$G$. Consider the case when the directed expression adds an arc~$(u,v)$ instead of an edge~$\set{u,v}$ of~$G$. Since the expression is canonical, we know that the opposite arc~$(v,u)$ must be added before the next union or relabel operation occurs. Thus, both arc operations must affect the same set of vertices and can therefore be combined into one edge operation~$\cwedge_{\ell(u),\ell(v)}$, where~$\ell(u)$ and~$\ell(v)$ are the labels of~$u$ and~$v$ at this point of the construction. As the replaced arc operations cannot add any arcs of~$G$, due to there being no opposite arcs in~$G$, this results in a mixed $2\ell$-expression constructing~$G$.

  A mixed graph~$\tilde{G}$ such that~$\cwm(\tilde{G})\neq \cwd(\tilde{G})$ is the partial orientation of the star~$K_{1,2}$,
  where the center vertex is connected to one leaf by an edge and to the other by an outgoing arc.
  To construct this graph with a mixed expression, we need at least three labels, as we cannot add edges or arcs between vertices of the same label, and we need to be able to distinguish each vertex. However, the directed variant of this graph has directed cliquewidth~$2$. We can first introduce the center vertex with label~$1$ as well as the leaf with the incident edge with label~$2$. We then add the arc going into the center vertex. Following this, we introduce the other leaf with label~$2$ and union both subgraphs. 
  Lastly, we add both arcs going from the center vertex to the leaves with one arc operation, as this operation does not need to distinguish between the two leaves.
\end{proof}

\NDMEquivRelation*
\label{lem:ndm-equiv-relation*}
\begin{proof}
  Clearly,~$\sim_\ndm$ is reflexive and symmetric. To show that it is transitive, let~$u,v,w$ be vertices of~$G$ such that~$u\sim_\ndm v$ and~$v\sim_\ndm w$. The directed neighborhoods are simple, as~$\inneigh(u)=\inneigh(v)=\inneigh(w)$ and~$\outneigh(u)=\outneigh(v)=\outneigh(w)$ follow directly. Note that, as we do not allow loops and opposite arcs, there cannot be any arc between~$u$,~$v$, and~$w$.
  For the undirected neighborhoods, it follows from the definition that ${\uneigh(u)\setminus\set{v,w}=\uneigh(v)\setminus\set{u,w}=\uneigh(w)\setminus\set{u,v}}$. Thus, it remains to show that if~$v$ is adjacent to one of~$u,w$ it is also adjacent to the other.
  If~$\set{u,v}\in E(G)$, it holds that~$u\in\uneigh(v)$. Since~$\uneigh(v)\setminus\set{w}=\uneigh(w)\setminus\set{v}$, it follows that~$u\in\uneigh(w)$ and thus that~$\set{u,w}\in E(G)$. It follows that~$w\in\uneigh(u)$ and thus that~$w\in\uneigh(v)$, yielding~$\set{v,w}\in E(G)$. It follows that $\uneigh(u)\setminus\set{v,w}=\uneigh(w)\setminus\set{u,v}$ is equivalent to~$\uneigh(u)\setminus\set{w}=\uneigh(w)\setminus\set{u}$, and thus~$u\sim_\ndm w$. 
  Therefore, the relation~$\sim_\ndm$ is transitive and vertices of the same type induce a clique or an independent set.
\end{proof}

\NDMtoCWM*
\label{prop:ndm-cwm*}
\begin{proof}
  We show how, given a mixed graph~$G$ with types~$\set{V_1,\dots,V_w}$, we can construct a mixed $(w+1)$-expression constructing~$G$. The labels~$1$ to~$w$ are used to construct the~$w$ subgraphs induced by the different types, while the label~$w+1$ is used as an auxiliary label to construct cliques induced by types. If a type~$V_i$ induces an independent set, we can simply introduce all vertices of~$V_i$ with label~$i$. 
  If a type~$V_i$ induces a clique, we create~$G[V_i]$ by repeating the following steps for each vertex in~$V_i$: introduce the vertex with label~$w+1$, connect the new vertex with all previously added vertices of type~$i$ (using the operation~$\cwedge_{i,w+1}$), and lastly relabeling the new vertex to~$i$ (using the operation~$\cwrelabel_{{w+1}\to i}$).
  After constructing all subgraphs induced by the types, we can union them and add all remaining edges and arcs between types using the respective operations, as at this point vertices of the same type have the same label.
\end{proof}

\CWMtoNDUIncomp*
\label{prop:cwm-ndu-incomp*}
\begin{proof}
  As $\cwu$ does not bound $\ndu$ neither does $\cwm$. Further, the graphs~$G_\ell$ constructed in the proof of \cref{prop:cwm-cwu} have as underlying undirected graph a complete graph, thus having~$\ndu=1$, but~$\cwm(G_\ell)\geq \ell+1$.
\end{proof}

\ParamTCUpperBound*
\label{prop:param-tc-upper-bound*}
\begin{proof}
  For any of these parameters~$\alpha$, it is well known and easy to see that their values do not increase on subgraphs. 
  As the underlying undirected graph of~$G$ is a subgraph of the underlying undirected graph of~$G^+$, it follows that~$\alpha(G)\leq \alpha^+(G)$.

  To show that the parameters do not bound their transitive closure variants, consider the family of graphs~$(G_\ell)_{\ell\geq 1}$ where~$G_\ell$ consists of vertices~$u_1,\dots,u_\ell,v,w_1,\dots,w_\ell$ and arcs~$(u_i,v),(v,w_i)$ for~$i\in[\ell]$. Thus, the underlying undirected graph of~$G_\ell$ is the star~$K_{1,2\ell}$ with center~$v$. It holds that~$\vc(G_\ell)=1$ and, since vertex cover bounds all of these parameters, all the other parameters are also constant. However, in the underlying undirected graph of the transitive closure, the subgraph induced by the $u$'s and $w$'s is the complete bipartite graph~$K_{\ell,\ell}$, which has treewidth at least~$\ell$. Therefore, as treewidth is a lower bound for all of these parameters, it follows that all these parameters are at least~$\ell$ on the transitive closure.
\end{proof}

\CWMDirCompGraph*
\label{claim:cwm-dir-comp-graph*}
\begin{claimproof}
  An acylic tournament graph contains a directed Hamiltonian path~$\angles{v_1,\dots,v_\ell}$ with each vertex~$v_i$ having incoming arcs from all previous vertices~$v_j$ with~$j<i$. In the following, we present a mixed $2$-expression that incrementally constructs the graph. 
  First, the vertex~$v_1$ is introduced with label~$1$. 
  For each subsequent vertex~$v_i$, we proceed as follows: We first introduce~$v_i$ with label~$2$. Then, we connect~$v_i$ to all preceding vertices, which are all labeled~$1$, using the operation~$\cwarc_{1,2}$. Finally, we relabel~$v_i$ to~$1$.
  This results in a mixed $2$-expression constructing the acylic tournament graph.
\end{claimproof}

\NDMCWMUpperBoundTC*
\label{prop:ndm/cwm-upper-bound-tc*}
\begin{proof}
  We first show the claim for neighborhood diversity, by showing that two vertices of the same type in~$G$ remain of the same type in~$G^+$. Suppose that two vertices~$u,v$ are of the same type in~$G$, but not in~$G^+$. As the neighborhood of a vertex can only increase in~$G^+$ compared to~$G$ by the addition of arcs, there must be a vertex~$w$ which is adjacent to~$u$ but not to~$v$ in~$G^+$. Assume w.l.o.g.\ that~$w\in\outneigh_{G^+}(u)$. Therefore, there is a directed path~$\angles{u,x,\dots,w}$ in~$G$. As~$u$ and~$v$ are of the same type in~$G$, it follows that there is also the directed path~$\angles{v,x,\dots,w}$ in~$G$, and thus~$w\in\outneigh_{G^+}(v)$, a contradiction.

  For cliquewidth, we show in the following how to adapt a given mixed $\ell$-expression constructing~$G$ into a mixed $(4^\ell\cdot \ell)$-expression constructing~$G^+$. 
  We proceed inductively on the given mixed expression. We label each vertex~$v$ using a new composite label~$(i,I,O)$, where~$i$ is the old label and~$I$ and~$O$ are subsets of the original labels. We set~$I$ and~$O$ such that it holds before every union operation that~$I$ is the set of original labels (of vertices) from which there is a directed path to~$v$ and~$O$ is the set of original labels (of vertices) to which there is a directed path from~$v$ (with the original label~$i$ of~$v$ included in both~$I$ and~$O$). 
  Therefore, when introducing a vertex with label~$i$, we instead introduce said vertex with label~$(i,\set{i},\set{i})$. 
  Any union operation remains unchanged, as it does not modify any labels.
  Any relabel operation~$\cwrelabel_{i\to j}$ is replaced by the relabel operation~$\cwrelabel_{(i,I,O)\to (j,I,O)}$ for each possible choice of~$I$ and~$O$. Further, we also have to update all labels where~$i$ is contained in the second and third component. Thus, we replace each label~$(x,I,O)$ with~$i\in I$ using the relabel operation~$\cwrelabel_{(x,I,O)\to (x,I\setminus\set{i}\cup\set{j},O)}$, and each label~$(x,I,O)$ with~$i\in O$ using the relabel operation~$\cwrelabel_{(x,I,O)\to (x,I,O\setminus\set{i}\cup\set{j})}$.
  Any edge operation~$\cwedge_{i,j}$ is replaced by the edge operation~$\cwedge_{(i,I,O),(j,I',O')}$ for each possible choice of~$I$,~$O$,~$I'$, and~$O'$.

  Similarly, any arc operation~$\cwarc_{i,j}$ is replaced by the arc operation~$\cwarc_{(i,I,O),(j,I',O')}$ for each possible choice of~$I$,~$O$,~$I'$, and~$O'$ (note that not all of these operations will add arcs, as it may happen that there is no vertex with such a label). As the added arcs create new directed paths, we also have to update the labels accordingly. 
  Further, we have to add the transitive arcs. For each pair of labels~$(x,I,O)$ and~$(y,I',O')$ such that~$i\in O$ and~$j\in I'$, we add the arc operation~$\cwarc_{(x,I,O),(y,I',O')}$, as there is now a directed path from vertices with label~$(x,I,O)$ to vertices with label~$(y,I',O')$.

  Currently, these additional arc operations may add arcs multiple times and result in arcs being parallel to edges. However, as shown by Courcelle and Olariu~\cite{CourcelleOlariu00BoundsCliqueWidth}, this can be avoided by removing redundant arc and edge operations, without increasing the number of labels.
  Thus, we obtain a mixed expression using at most~$4^\ell\cdot \ell$ labels that constructs~$G^+$.
\end{proof}

\ChiUpperBound*
\label{thm:chi-upper-bound*}
\begin{proof}
  Given a mixed graph~$G$ and its layering~$\angles{L_0,L_1,\ldots,L_{\maxrank}}$, we obtain a proper coloring~$c$ as follows: 
  We color each layer~$L_i$ with an optimal coloring of the undirected graph~$G[L_i]$ using colors from the set~$\set{\sum_{j=0}^{i-1}\chi_\mathrm{u}(G[L_j])+1,\ldots,\sum_{j=0}^{i}\chi_\mathrm{u}(G[L_j])}$.
  As each arc is oriented towards the higher layer, this ensures that for every arc~$(u,v)$ (and also every edge~$\set{u,v}$) with~$u\in L_i$ and~$v\in L_j$ that~$c(u)<c(v)$ since~$i<j$. Furthermore, for each edge~$\{u,v\}$ with~$u,v\in L_i$ it holds that~$c(u)\neq c(v)$, as we used a proper coloring of~$G[L_i]$. Thus, $c$ is a proper coloring of~$G$ using~$\sum_{i=0}^{\maxrank}\chi_\mathrm{u}(G[L_i])$ colors.

  To show that this bound is tight, we construct a family of mixed graphs~$(G_{\ell,k})_{\ell,k\geq 1}$ such that~$\maxrank(G_{\ell,k})=\ell$,~$\chi_\mathrm{u}(G_{\ell,k}[L_i])=k$ for each layer~$L_i$, and~$\chi(G_{\ell,k})=(\ell+1)\cdot k$. We achieve this by constructing~$G_{\ell,k}$ from~$\ell+1$ copies of the complete graph~$K_k$ which form its layers, where each vertex in layer~$L_i$ has an outgoing arc to every vertex in layer~$L_{i+1}$ for~$i\in[\ell-1]$. Thus, each layer has chromatic number~$k$, and as the arcs enforce that the color intervals used by each layer are disjoint, the chromatic number of the entire graph is~$(\ell+1)\cdot k$.
\end{proof}

\RankBounds*
\label{prop:rank-bounds*}
\begin{proof}
  The result for~$\ndm^+$ follows directly from \cref{lem:ndm-dir-path} and the fact that~$\maxrank(G)=\maxrank^+(G)$.

  For $\vc$ and $\td$ and~$\tw$, it is well known and easy to see that their value on a subgraph is a lower bound for the value on the entire graph. Further, it holds for a path~$P_\ell$ of length~$\ell$ that~$\vc(P_\ell)=\ceil{\ell/2}$ and~$\td(P_\ell)=\ceil{\log_2(\ell+2)}$. As a mixed graph with maxrank~$\ell$ contains a directed path of length~$\ell$ as a subgraph, the bounds for~$\vc$ and~$\td$ follow.
  Regarding~$\tw^+$, we know that a directed path of length~$\ell$ leads to a clique of size~$\ell+1$ in the underlying undirected graph of the transitive closure. As a clique of size~$\ell+1$ has treewidth~$\ell$, the bound follows.
\end{proof}

\ChiBounds*
\label{cor:chi-bounds*}
\begin{proof}
  It is well known that~$\chi_\mathrm{u}(G)\leq \td(G)-1$,~$\chi_\mathrm{u}(G)\leq \vc(G)+1$, and~$\chi_\mathrm{u}(G)\leq\tw(G)+1$. Together with \cref{thm:chi-upper-bound} and \cref{prop:rank-bounds}, the bounds follow.

  Regarding the second bound:
  We construct a proper $(2\vc(G)+1)$-coloring~$c$ of~$G$ given a vertex cover~$C$ of size~$\vc(G)$ as follows: We first sort the vertices of the cover~$C$ topologically (according to the arcs between them in the transitive closure~$G^+$, i.e., we sort~$G^+[C]$ topologically), resulting in an ordering~$\angles{v_1,\dots,v_{\vc(G)}}$. We then color these vertices with distinct colors, with each vertex~$v_i$ obtaining the color~$2i$. This yields a proper coloring of~$G^+[C]$.
  For any remaining uncolored vertex~$v\in V(G)\setminus C$ it holds that~$v$ only has neighbors in~$C$, and thus that the neighbors of~$v$ are already colored. Let $c^-(v)=\max_{u\in N^-(v)}c(u)$ and~$c^+(v)=\min_{u\in N^+(v)}c(u)$ be the biggest and smallest color of the incoming and outgoing neighbors of~$v$, respectively. In the case that~$v$ has no incoming or no outgoing neighbors we set~$c^-(v)=0$ and~$c^+(v)=2\vc(G)+2$, respectively.     
  It holds that~$c^-(v)<c^+(v)$. Indeed, suppose that~$c^+(v)\geq c^-(v)$. Then there is a~$u^+\in N^+(v)$ with~$c(u^+)=c^+(v)$ and a~$u^-\in N^-(v)$ with~$c(u^-)=c^-(v)$. It follows, that there are arcs~$(u^-,v)$ and~$(v,u^+)$ in~$A(G)$, and thus there is the arc~$(u^-,u^+)\in A(G^+)$. As~$u^-,u^+\in C$, it follows that $c^-(v)=c(u^-)<c(u^+)=c^+(v)$ due to our initial coloring being according to the topological order of~$G^+[C]$, a contradiction to~$c^-(v)\geq c^+(v)$. Thus, it holds that~$c^+(v)-c^-(v)\geq 2$, due to all colors of vertices in~$C$ being even. We therefore set~$c(v)=c^-(v)+1$. Now all incoming neighbors have a smaller color, all outgoing neighbors have a bigger color, and all remaining neighbors (adjacent via edges) are in~$C$ and have thus an even color, which cannot be~$c(v)$, as~$c(v)$ is uneven. Thus, the coloring remains proper. By coloring all vertices in~$V(G)\setminus C$ in this manner, we obtain a proper coloring of~$G$ with~$2\vc(G)+1$ colors.
  See \cref{fig:path-color-vc} for an example of the obtained coloring.

  \begin{figure}[tbp]
      \centering
      \includegraphics{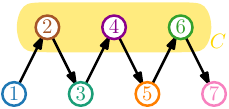}
      \caption{The path~$P_6$ with a smallest vertex cover~$C$ of size~$3$ and resulting (optimal) $7$-coloring.}
      \label{fig:path-color-vc}
  \end{figure}

  It remains to show that the bound is tight. For any~$\ell\in\N$, the directed path~$P_{2\ell}$ of length~$2\ell$ has chromatic number~$2\ell+1$. Furthermore, the smallest vertex cover contains every second vertex and has thus a size of~$\ell$.
\end{proof}

\LogicFPTCWK*
\label{thm:logic-fpt-cw-k*}
\begin{proof}
  In \emph{monadic second-order logic on mixed graphs}, we can construct expressions using the following primitives: 
  \begin{itemize}
    \item variables for vertices $(x,y)$ and sets of vertices $(X)$
    \item binary predicates for equality $(=)$, membership $(\in)$, as well as edge and arc predicates $E(x,y)$ and $A(x,y)$, which are true if there is an edge or arc from vertex~$x$ to vertex~$y$, respectively
    \item standard logic operators: $\lnot,\land,\lor,\rightarrow$
    \item quantifiers: $\forall$ and $\exists$
  \end{itemize}
  \newcommand{\fmpartition}{\mathsf{partition}}
  \newcommand{\fmcheckarcs}{\mathsf{checkArcs}}
  \newcommand{\fmcheckedges}{\mathsf{checkEdges}}
  \newcommand{\fmpropercoloring}{\mathsf{properColoring}}
  \newcommand{\fmcolorable}{\mathsf{colorable}}

  To express \MixedColoring{}, we first construct the formula~$\fmpartition$ that checks if a given collection of vertex sets~$X_1,\dots,X_k$ corresponds to the color classes of a proper coloring. This formula can be expressed as follows~\cite{Barr20Courcelle}:
  \begin{gather*}
    \fmpartition(X_1,\dots,X_k):=\forall x\left[\bigvee_{i=1}^k (x\in X_i)\right]\land\lnot\exists x\left[\bigvee_{i\neq j}^k\left(x\in X_i\land x\in X_j\right)\right]
  \end{gather*}
  We further need to check that no arcs or edges are violated, which we achieve using the following formulas:
  \begin{gather*}
    \fmcheckarcs(X_1,\dots,X_k):=\forall x,y\left[A(x,y)\rightarrow\lnot\left(\bigvee_{i\geq j}^k (x\in X_i\land y\in X_j)\right)\right]\\
    \fmcheckedges(X_1,\dots,X_k):=\forall x,y\left[E(x,y)\rightarrow\lnot\left(\bigvee_{i=1}^k (x\in X_i\land y\in X_i)\right)\right]
  \end{gather*}
  Combining these formulas, we construct a formula that checks if a given collection of vertex sets~$X_1,\dots,X_k$ corresponds to the color classes of a proper mixed coloring:
  \begin{align*}
    \fmpropercoloring(X_1,\dots,X_k):=\fmpartition(X_1,\dots,X_k)&\land\fmcheckarcs(X_1,\dots,X_k)\\
    &\land\fmcheckedges(X_1,\dots,X_k)
  \end{align*}
  By existentially quantifying over the vertex sets, we obtain a formula that checks if the graph is $k$-colorable:
  \begin{gather*}
    \fmcolorable:=\exists X_1,\dots,X_k\colon\fmpropercoloring(X_1,\dots,X_k)
  \end{gather*}

  Using the version of Courcelle's theorem shown by Arnsborg et al.~\cite{ArnborgLS91MixedCourcelle}, we obtain that \MixedColoring{} is \FPT{} parameterized by treewidth plus chromatic number.

  To obtain the result for mixed cliquewidth, we transform a given mixed graph~$G$ into the directed graph~$D(G)$ as introduced before \cref{prop:cwm-tw}. 
  Instead of the predicates~$E(x,y)$ and~$A(x,y)$, we can now only use an arc predicate~$D(x,y)$ that is true if there is an arc from~$x$ to~$y$ in~$D(G)$.
  We therefore have to adapt our expression to work on the transformed graph. Since we only consider mixed graphs without directed cycles, and thus without opposite arcs, it holds that~$D(G)$ contains opposite arcs between a pair of vertices~$x,y$ if and only if there is an edge between~$x$ and~$y$ in~$G$. Thus, we can substitute our predicates as follows: $E(x,y)$ is substituted by $D(x,y)\land D(y,x)$ and $A(x,y)$ is substituted by $D(x,y)\land\lnot D(y,x)$. 
  Applying Courcelle's theorem for directed graphs~\cite{CourcelleEngelfriet12GSaMSOL} yields that \MixedColoring{} is \FPT{} parameterized by the directed cliquewidth of~$G$ plus chromatic number. As mixed and directed cliquewidth are equivalent,~\cref{prop:cwm-tw}, we obtain the desired result.
\end{proof}

\ListColoringWOneVC*
\label{thm:list-coloring-w1-vc*}
\begin{proof}
    The problem \PMulticoloredClique is defined as follows.
    \defdecproblem{\PMulticoloredClique}
    {Graph~$G$ together with an $\ell$-coloring.}
    {Is there an $\ell$-clique containing exactly one vertex of each color?}

    Given an instance~$(G,c)$ of \textsc{MulticoloredClique}, Fellows et al.~\cite{FellowsFLRSST11ComplexColorParaTW} construct an instance $(G',L')$ of \ListColoring{} as follows: 
    For each color class~$V_i$ of~$c$, they create a \emph{color-class-vertex}~$v_i$ in~$G'$ and set its list~$L(v_i)$ to $V_i$.\footnote{The vertices of~$G$ can be treated as colors by numbering them arbitrarily. For simplicity, we use the vertices as colors, without explicit numbering.} For every pair of differently colored non-adjacent vertices~$x,y$, i.e., for~$x\in V_i$ and~$y\in V_j$ such that~$i\neq j$ and~$\set{x,y}\not\in E(G)$, they add an \emph{edge-vertex}~$v_{xy}$ to~$G'$ and set it adjacent to~$v_i$ and~$v_j$ with~$L(v_{xy})=\set{x,y}$. Since the color-class-vertices~$v_i$ form a vertex cover of~$G'$, the vertex cover number is at most~$\ell$. An example of the reduction can be seen in \cref{fig:list-color-multicolored-clique}.

    \begin{figure}[tbp]
        \centering
        \begin{subfigure}[b]{0.45\textwidth}
            \centering
            \includegraphics{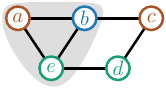}
            \caption{}
            \label{fig:multicolored-clique}
        \end{subfigure}
        \hspace{0.05\textwidth}
        \begin{subfigure}[b]{0.45\textwidth}
            \centering
            \includegraphics{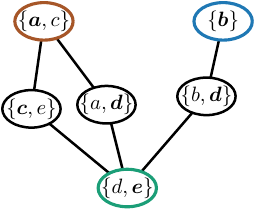}
            \caption{}
            \label{fig:list-color-from-multicolored-clique}        
        \end{subfigure}
        \caption{(a) An instance of \textsc{MulticoloredClique} with three colors and five vertices. It has a multicolored $3$-clique consisting of the vertices~$a$,~$b$, and~$e$. 
        (b) The instance of \ListColoring{} resulting from the \textsc{MulticoloredClique} instance. The vertices with a bold (colored) rim are the color-class-vertices corresponding to the color classes of the \textsc{MulticoloredClique} instance. 
        Since the \textsc{MulticoloredClique} instance has a multicolored $3$-clique, the \ListColoring{} instance has a proper list coloring (highlighted in bold).}
        \label{fig:list-color-multicolored-clique}
    \end{figure}

    We now show that a multicolored $\ell$-clique in~$G$ corresponds to a proper list coloring of~$G'$. 
    Given a multicolored $\ell$-clique~$C$, let~$\angles{w_1,\dots,w_\ell}$ be the vertices in~$C$ and assume w.l.o.g. that~$c(w_i)=i$. We color each color-class-vertex~$v_i\in G'$ with the color~$w_i$. For each edge-vertex~$v_{xy}$ it holds that~$\set{x,y}\not\in E(G)$. Thus,~$x$ and~$y$ cannot be both contained in~$C$. Therefore, as we treat vertices from~$G$ as colors in~$G'$, at least one of these two vertices can be used to color~$v_{xy}$, resulting in a proper list coloring of~$G'$.

    Given a proper list coloring~$c'$ of~$G'$, the set~$C=\set{c'(v_i)\mid i\in[\ell]}$ is a multicolored $\ell$-clique. Indeed, suppose that there are~$x,y\in C$ with~$\set{x,y}\not\in E(G)$. Let~$c(x)=i$ and~$c(y)=j$. Then,~$v_{xy}$ is adjacent to~$v_i$ and~$v_j$ which are colored~$x$ and~$y$, respectively. As the color of~$v_{xy}$ must be either~$x$ or~$y$, the coloring~$c'$ would be improper, a contradiction.
\end{proof}

\ParaListColoringToMixedColoring*
\label{lem:para-list-coloring-to-mixed-coloring*}
\begin{proof}
    Given an instance~$(G,L)$ of \ListColoring{} with color list~$L(v)\subset [\ell]$ for each vertex~$v$ of~$G$, we construct the following mixed graph~$G'$: For each vertex~$v$ of $G$ and~$j\in[\ell]-L(v)$, we add the directed path~$\angles{w^{v,j}_1,\dots,w^{v,j}_{\ell}}$, denoted~$P^{v,j}_{\ell-1}$. For each color~$j\in [\ell]-L(v)$, we add the edge~$\set{v,w^{v,j}_j}$. See \cref{fig:list-to-mixed-tw} for an example of the additions for each vertex.
    In any proper $\ell$-coloring of~$G'$, a vertex~$w^{v,j}_i$ must be colored with the color~$i$, as it is part of a directed path of length~$\ell-1$. As~$v$ is adjacent to vertices~$w^{v,j}_j$ for~$j\in[\ell]-L(v)$, it must subsequently be colored with a color in~$L(v)$. Thus, a proper~$\ell$-coloring of~$G'$ corresponds to a proper list coloring of~$G$ and vice versa.

    \begin{figure}[tbp]
        \centering
        \includegraphics{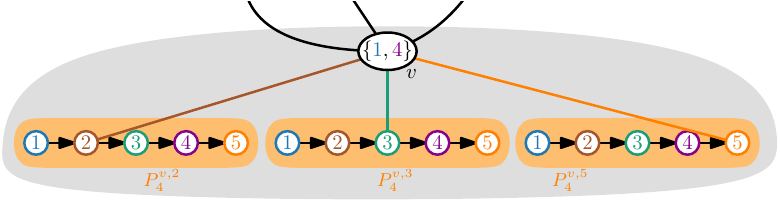}
        \caption{An example of the additions for a vertex~$v$ in an instance of \ListColoring{} with~$\ell=5$ and~$L(v)=\set{1,4}$. Each added path~$P^{v,j}_{4}$ prevents~$v$ from being colored with the color~$j$ in a proper $5$-coloring.}
        \label{fig:list-to-mixed-tw}
    \end{figure}

    As we can assume w.l.o.g.\ that~$\ell\leq n^2$, where~$n$ is the number of vertices in~$G$, the reduction is possible in polynomial time.
    It remains to show that the parameters remain bounded. We only modified~$G$ by adding paths, each connected via a single edge to a vertex in~$G$. Thus, no new cycles were created, and feedback vertex set remains the same.

    To show that pathwidth remains bounded using this reduction, we first show that we can perform a slightly different reduction that gives a simpler argument for pathwidth remaining bounded. Instead of adding the path~$P^{v,j}_{\ell-1}$ for each forbidden color~$j\in[\ell]-L(v)$, we add the path~$P^{v}_{\ell-1}$ once for each vertex~$v$, and connect it to~$v$ via the edges~$\set{v,w^{v}_j}$ for each forbidden color~$j\in[\ell]-L(v)$. This still enforces the color list of~$v$, and is thus a valid reduction.
    In the following, we show that this reduction increases the pathwidth by at most~$2$. 
    Given a path-decomposition~$\mathcal{T}$ of~$G$ of minimal width~$\tw(G)$, we construct a path-decomposition~$\mathcal{T}'$ of $G'$ as follows: For each vertex~$v$ of~$G$, let~$B_t$ be a bag containing~$v$ that contains at most~$\tw(G)+1$ vertices. To obtain~$\mathcal{T}'$, we replace the node~$t$ with a path~$\angles{t_0,t_1,\dots,t_{\ell-1},t_\ell}$. 
    Let~$\angles{t^{v}_{1},\dots,t^{v}_{\ell-1}}$ be the nodes of the path-decomposition of~$P^{v}_{\ell-1}$ with bags~$B_{t^{v}_i}=\set{w^{v}_i,w^{v}_{i+1}}$ for~$i\in[\ell-1]$.
    We set~$B_{t_0}=B_t$, $B_{t_i}=B_t\cup B_{t^{v}_i}$ for~$i\in[\ell-1]$, and~$B_{t_\ell}=B_t$. As~$P^{v}_{\ell-1}$ is only connected to~$v$, which is contained in~$B_t$, this results in a valid path-decomposition of~$G'$. As the new bags contain at most two additional vertices, the width of the path-decomposition of~$G'$ is at most the width of the path-decomposition of~$G$ plus~$2$.
    For the original reduction, where for each vertex we added several adjacent paths, we can apply the same argument to each added path, resulting in an overall increase of pathwidth by at most~$2$.
\end{proof}

\MixedColoringParaNPCW*
\label{thm:para-np-cwm-rank*}
\begin{proof}
  A string~$S^*$ is a \emph{superstring} of a string~$S$ if~$S$ is obtained from~$S^*$ through the deletion of characters.
  Given a set~$\mathcal{S}$ of strings over the binary alphabet~$\set{0,1}$ and an integer~$k$ the \PShortestSuperstring problem asks whether there is a superstring~$S^*$ of length at most~$k$ that is a superstring of each string in~$\mathcal{S}$. This problem was shown to be \NP-complete by Räihä and Ukkonen~\cite{RaihaUkkonen81SCSBinaryNPComplete}.

  Given a set of strings~$\mathcal{S}$ and an integer~$k$, we create the following mixed graph~$G$:
  For each string~$S\in\mathcal{S}$, we create a distinct directed path~$P_S$ that consists of so-called \emph{character vertices} that correspond to the characters in~$S$. 
  Further, we add edges between any two character vertices, belonging to different paths, that correspond to different characters. 
  The graph is constructed such that the colors of the character vertices in a proper $k$-coloring of~$G$ correspond to the positions of the characters in a superstring of~$\mathcal{S}$ of length~$k$. 
  See \cref{fig:sss-reduction*} for an example of the construction.

  \begin{figure}[tbp]
    \centering
    \captionsetup[subfigure]{justification=centering}
    \begin{subfigure}[b]{.45\textwidth}
      \centering
      \includegraphics{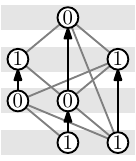}
      \caption{}
      \label{fig:sss-reduction-instance*}
    \end{subfigure}
    \begin{subfigure}[b]{.45\textwidth}
      \centering
      \includegraphics{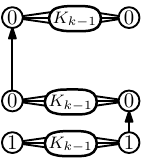}
      \caption{}
      \label{fig:sss-ladder*}
    \end{subfigure}
    \repeatcaption{fig:sss-reduction}{
      (a) The mixed graph resulting from an instance of \PShortestSuperstring with strings~$\set{01,100,11}$ and shortest superstring~$1010$.
      (b) The construction resulting from splitting the path corresponding to the string~$100$ at each vertex, and connecting the resulting vertices to distinct cliques to enforce that they receive the same color in a proper $k$-coloring.
    } 
    \label{fig:sss-reduction*}
  \end{figure}

  \proofsubparagraph{Correctness of Reduction}
  In the following, we show that~$G$ is $k$-colorable if and only if there is a superstring of~$\mathcal{S}$ of length~$k$. 
  We denote with~$v_S^i$ the $i$-th character vertex in~$P_S$. 

  Given a superstring~$S^*$ of~$\mathcal{S}$ of length~$k$, let~$p(S_i)$ denote the position of the $i$-th character of~$S$ in the superstring~$S^*$ for~$S\in\mathcal{S}$ and~$i\in[\abs{S}]$. 
  We color~$G$ by assigning for each~$S\in\mathcal{S}$ and~$i\in[\abs{S}]$ the color~$p(S_i)$ to the vertex~$v_S^i$, thereby using~$k$ colors in total. 
  Since consecutive characters of the same string must appear in increasing order in the superstring, no arc of~$P_S$ is violated. Further, all vertices of the same color must correspond to the same character as they correspond to the same position in the superstring and thus no edge is violated. Therefore, this is a proper $k$-coloring of~$G$.

  Conversely, given a proper $k$-coloring~$c$ of~$G$, we create a superstring~$S^*$ of~$\mathcal{S}$ of length~$k$ as follows: 
  Let~$t_i$ for~$i\in[k]$ be the unique character (i.e., $0$ or~$1$) that corresponds to the vertices colored with color~$i$. As we know that vertices corresponding to different characters cannot receive the same color, due to the edges, this is well-defined. This results in the string~$t_1,\dots,t_k$, which we call~$S^*$. It remains to show that~$S^*$ is a superstring of~$\mathcal{S}$. For each string~$S\in\mathcal{S}$, each character~$S_i$ of~$S$ appears at position~$c(v_S^i)$ in~$S^*$. Since the coloring is proper and~$P_S$ is directed from the first to the last character, it holds that~$c(v_S^i)<c(v_S^{i+1})$ for each~$i\in[\abs{S}-1]$, and thus the characters of~$S$ appear in order in~$S^*$. Therefore,~$S^*$ is a superstring of~$\mathcal{S}$.
  
  \proofsubparagraph{Constant Rank}
  In~$G$, every vertex has at most one incoming and one outgoing arc. 
  To obtain a mixed graph~$G'$ with maxrank~$1$ we modify~$G$ as follows: We split each character vertex into two \emph{siblings}, with one sibling having the incoming arc of the original vertex and the other sibling having the outgoing arc (if the character vertex is the first or last in its corresponding path~$P_S$ only one sibling has an incident arc). Further, between each pair of sibling vertices we add a distinct $(k-1)$-clique, with each of the siblings being adjacent (via edges) to the whole clique, thus ensuring that the siblings receive the same color in any proper $k$-coloring. Therefore, the resulting graph~$G'$ is $k$-colorable if and only if~$G$ is $k$-colorable.

  \proofsubparagraph{Constant Cliquewidth}
  It remains to show that~$G'$ has constant mixed cliquewidth. We show this in two steps, first showing how to construct each split path corresponding to a string~$S\in\mathcal{S}$ using~$6$ labels and then how to construct the whole graph using~$6$ labels.
  \begin{claim}
    Each split path corresponding to a string~$S\in\mathcal{S}$ can be constructed by a mixed expression using at most~$6$ labels such that in the resulting graph all vertices corresponding to $(k-1)$-cliques have label~$2$ and all character vertices have label~$0$ or~$1$, depending on which character they correspond to.
  \end{claim}
  \begin{claimproof}
    It is easy to see that a clique can be constructed using two labels. Thus, we can assume that we are given an expression that uses two labels and constructs a $(k-1)$-clique with all vertices being labeled~$3$ (at the end). We iteratively construct the split path as follows: For each pair of sibling vertices, we call the sibling with an incoming arc the \emph{incoming sibling} and the sibling with an outgoing arc the \emph{outgoing sibling}. We use the label~$\ell_i$ to introduce an incoming sibling and~$\ell_o$ to introduce an outgoing sibling (we could for example set~$\ell_i=4$ and~$\ell_o=5$). 
    We assume that we already have constructed the path up to the $i$-th character vertex, and proceed for the next pair of siblings as follows: We introduce the incoming sibling with label~$\ell_i$ and add the arc from the previous vertex with label~$\ell_o$ to it (this step is omitted for the first character vertex, which has no incoming arc). Since the previous vertex has no more connections to be added (in this path), we relabel it to its corresponding character. We then add the $(k-1)$-clique with label~$3$ and connect it to the incoming sibling. As the incoming sibling has no more connections to be added, we relabel it to its corresponding character. We then introduce the outgoing sibling with label~$\ell_o$ and connect it to the $(k-1)$-clique. Following this, the $(k-1)$-clique has no more connections to be added, so we relabel it to~$2$. The outgoing sibling is relabeled to its corresponding character after adding the connection to the incoming sibling of the next character vertex (except if it is itself the last character vertex, then it is relabeled to its corresponding character immediately). In total, we used $6$ labels: two for the clique, two for the siblings, and two for the corresponding characters of character vertices.
  \end{claimproof}
  \begin{claim}
    The whole graph~$G'$ can be constructed by a mixed expression using at most~$6$ labels and has mixed cliquewidth at most~$6$.
  \end{claim}
  \begin{claimproof}
    We construct the graph iteratively, adding the split paths one by one. Before adding a new path to the existing graph, we first relabel all the vertices of the new path to ensure that the labels of the new path are distinct from the labels of the existing graph, to enable the addition of edges between the new path and the existing graph. This leads to us using six labels, as a split path consists of three labels. After adding the new path to the existing graph via a union operation, we add edges between vertices corresponding to different characters. Finally, we relabel the labels of the new path back to the original labels, so that the next path can be added in the same manner. 
    As each split path is constructed using at most~$6$ labels, and we only use~$6$ labels in total to combine all paths, the mixed cliquewidth of the whole graph is at most~$6$.
  \end{claimproof}

  Thus, the resulting graph~$G'$ is $k$-colorable if and only if there is a superstring of~$\mathcal{S}$ of length~$k$, has maxrank~$1$, and has mixed cliquewidth at most~$6$. Therefore, \MixedColoring{} is \paraNP-hard parameterized by mixed cliquewidth plus maxrank.
\end{proof}

\CorrectnessPCS*
\label{claim:correctness-pcs*}
\begin{claimproof}
  Let~$\sigma$ be a schedule of length~$D$ starting at time~$0$. We denote with~$\sigma(t)$ the starting time of task~$t$.
  We create a $4D$-coloring of the resulting graph by assigning for each task~$t\in T_1$ the color~$4\sigma(t)+1$ to~$v_t^-$ and the color~$4\sigma(t)+3$ to~$v_t^+$, while for each task~$t\in T_2$ we assign the color~$4\sigma(t)+2$ to~$v_t^-$ and the color~$4\sigma(t)+4$ to~$v_t^+$. Furthermore, we assign the (directed) path~$P$ its unique $4D$-coloring. As at each point in time, at most one task is scheduled on each machine, no two task vertices receive the same color. Additionally, we colored the task vertices according to their types, ensuring that no edge between a task vertex and the path~$P$ is violated. Lastly, as the schedule respects the precedence constraints, it holds for each arc~$(v_t^+,v_{t'}^-)$ that~$\sigma(t)+1\leq \sigma(t')$, and thus~$v_t^+$ receives a smaller color than~$v_{t'}^-$. This results in all arcs being respected. Therefore, we have obtained a proper $4D$-coloring of the resulting graph.

  Conversely, given a proper $4D$-coloring~$c$ of the resulting graph, we create a schedule~$\sigma$ by assigning each task~$t\in T_1 \cup T_2$ the start time~$\sigma(t)=\floor{c(v_t^-)/4}$. As the task vertices were colored according to their types, there can be at most two tasks with the same start time, one on each machine. 
  It remains to show that the schedule respects the precedence constraints. For two tasks~$t\prec t'$ it holds due to the arcs that~$c(v_t^-)<c(v_t^+)<c(v_{t'}^-)$. As~$c(v_t^+)$ is~$3$ or~$4$ modulo~$4$ and~$c(v_{t'}^-)$ is~$1$ or~$2$ modulo~$4$, it must hold that~$\floor{c(v_t^+)/4}<\floor{c(v_{t'}^-)/4}$, and thus~$\sigma(t)+1\leq\sigma(t')$.
  Therefore, the schedule respects the precedence constraints and is a valid schedule of length~$D$.
\end{claimproof}

\NDTCPCS*
\label{claim:nd-tc-pcs*}
\begin{claimproof}
  Recall that by construction, the task vertices form a clique in the underlying graph of~$G$ and therefore also in the underlying graph of the transitive closure~$G^+$.
  Since~$P$ is a directed path, the vertices of~$P$ form a clique in the underlying graph of~$G^+$ as well.
  Thus, there are four types of task vertices, depending on their assigned machine and whether they are a start or end vertex, and there are four types of vertices on the path~$P$, depending on their position, as this determines to which of the task types they are (not) adjacent.
\end{claimproof}

\ProperPreorderColorings*
\label{claim:proper-preorder-colorings*}
\begin{claimproof}
  Given a coloring~$c$ that corresponds to~$p$, we know that there is a list of~$\ell$ ascending colors~$\angles{c_1,\dots,c_\ell}$ such that for each type~$C$ it holds that~$c(C)\subseteq[c_{p^-(C)}, c_{p^+(C)})$. As~$p$ is proper, for each arc~$(C_i,C_j)\in A(G)$ it holds that~$p^+(C_i)\leq p^-(C_j)$, and thus that $c_{p^+(C_i)}\leq c_{p^-(C_j)}$. Therefore, it holds for each pair of vertices~$v_i\in C_i$ and~$v_j\in C_j$ that~$c(v_i)<c_{p^+(C_i)}\leq c_{p^-(C_j)}\leq c(v_j)$, and thus no arc is violated.

  Given a proper coloring~$c$ of~$G$, we obtain a proper type-endpoint preorder~$p$ to which~$c$ corresponds as follows: 
  We denote for each type~$C$ with~$c^-(C)$ and~$c^+(C)$ the \emph{colored type-endpoints}, i.e., the largest and smallest color such that~$c(C)\subseteq[c^-(C),c^+(C))$. We obtain a proper type-endpoint preorder~$p$ by setting~$p^-(C)$ to~$i$ if~$c^-(C)$ is the $i$-th smallest colored type-endpoint (we proceed analogously for~$p^+(C)$). Clearly,~$c$ corresponds to this preorder~$p$. It remains to show that the preorder~$p$ is proper. For each arc~$(C_i,C_j)\in A(G)$, let~$v_i$ be the vertex with the largest color in~$C_i$ and let~$v_j$ be the vertex with the smallest color in~$C_j$ (under coloring~$c$). As we defined the colored type-endpoints to be as tight as possible, it holds that~$c^+(C_i)=c(v_i)+1$ and~$c^-(C_j)=c(v_j)$. Since~$c$ is a proper coloring, it holds that~$c(v_i)<c(v_j)$, and thus~$c^+(C_i)\leq c^-(C_j)$. Therefore, it holds that~$p^+(C_i)\leq p^-(C_j)$, and thus~$p$ is a proper preorder. 
\end{claimproof}

\ILPCorrectness*
\label{claim:ilp-correctness*}
\begin{claimproof}
  Recall the \ILP{} constructed for a given proper type-endpoint preorder~$p$ and a number of colors~$k$:
  
  \begin{subequations}
  \begin{align}
    c_i+1 & \leq c_{i+1} \quad&&\forall i\in [\ell-1]\label{constr:preorder-endpoints}\\
    \sum_{\mathcal{S}'\subseteq \mathcal{S}}x_{\mathcal{S}',i} & \leq c_{i+1}-c_i \quad&&\forall i\in [\ell-1]\label{constr:interval-size}\\
    \sum_{i=1}^{p^-(C)-1}\sum_{\mathcal{S}'\subseteq \mathcal{S}\colon C\in \mathcal{S}'}x_{\mathcal{S}',i} & =0 \quad&&\forall C\in \mathcal{S}\label{constr:lower-type-endpoint}\\
    \sum_{i=p^-(C)}^{p^+(C)-1}\sum_{\mathcal{S}'\subseteq \mathcal{S}\colon C\in \mathcal{S}'}x_{\mathcal{S}',i} & =\abs{C} \quad&&\forall C\in \mathcal{S}\label{constr:clique-size}\\
    \sum_{i=p^+(C)}^{\ell-1}\sum_{\mathcal{S}'\subseteq \mathcal{S}\colon C\in \mathcal{S}'}x_{\mathcal{S}',i} & =0 \quad&&\forall C\in \mathcal{S}\label{constr:upper-type-endpoint}\\
    x_{\mathcal{S}',i} & =0 \quad&&\forall i\in [\ell-1]\forall \mathcal{S}'\subseteq \mathcal{S}\colon \nonumber\\
    & &&\exists C,C'\in \mathcal{S}'\colon \set{C,C'}\in E(G)\label{constr:edge}\\
    c_i & \in [k] \quad&&\forall i\in [\ell]\label{bound:color-endpoints}\\
    x_{\mathcal{S}',i} & \in\set{0,\dots,k} \quad&&\forall i\in [\ell-1]\forall \mathcal{S}'\subseteq \mathcal{S}\label{bound:color-count}
  \end{align}
  \end{subequations}

  Given a proper $k$-coloring~$c$ of~$G$ that corresponds to the given proper preorder~$p$ via the list of colors~$\angles{c_1',\dots,c_\ell'}$, we can find a feasible solution for the \ILP{} by setting each variable~$c_i$ to~$c_i'$ and counting for each interval~$[c_i,c_{i+1})$ and each subset~$\mathcal{S'}\subseteq\mathcal{S}$ how many colors in the interval are used exclusively by the types in~$\mathcal{S'}$. Formally, we set~$x_{\mathcal{S'},i}=\abs{\set{d\in[c_i,c_{i+1})\mid \forall C\in \mathcal{S'}\colon d\in c(C)\land \forall C\notin \mathcal{S'}\colon d\notin c(C)}}$. This clearly satisfies all constraints.

  Conversely, given a feasible solution to the \ILP{}, we can construct a proper $k$-coloring~$c$ of~$G$ as follows: For each interval~$[c_i,c_{i+1})$, we iteratively assign the colors to the types by going through the subsets~$\mathcal{S'}$ of $\mathcal{S}$ in some arbitrary order and assigning each type in the subset~$\mathcal{S'}$ the next~$x_{\mathcal{S'},i}$ unassigned colors in the interval. When assigning colors to a type, we color some arbitrary uncolored vertices of the type. It remains to show that this results in a proper $k$-coloring of~$G$.
  First off, the constraint~\labelcref{constr:clique-size} ensures that each type gets assigned exactly~$\abs{C}$ colors to color all its vertices distinctly. Furthermore, the constraint~\labelcref{constr:edge} ensures that no two types (and therefore no two vertices) connected via an edge receive the same color. Lastly, for each arc~$(C_i,C_j)$ it holds that~$p^+(C_i)\leq p^-(C_j)$ (since~$p$ is a proper preorder) and due to the constraints~\labelcref{constr:lower-type-endpoint,constr:upper-type-endpoint} it holds that all colors assigned to vertices in~$C_i$ are smaller than the colors assigned to vertices in~$C_j$, thus respecting all arcs. Finally, the bounds ensure that at most~$k$ colors are used, resulting in~$c$ being a proper $k$-coloring of~$G$.
\end{claimproof}

\MISZColoringProp*
\label{claim:mis0-coloring-property*}
\begin{claimproof}
  Clearly, if~$G$ is empty, then~$\chi(G)=0$. Otherwise, we first show that~$\chi(G)\leq 1+\chi(G-I)$ for each maximal independent set~$I$ in the subgraph induced by the vertices of inrank~$0$. Let~$\angles{V_2,\dots,V_k}$ be the color classes of an optimal coloring of~$G-I$. Since~$I$ is an independent set and has no incoming arcs,~$\angles{I,V_2,\dots,V_k}$ forms the color classes of a proper~$k$ coloring of~$G$. Therefore, there exist proper colorings of~$G$ where the first color class forms a maximal independent set in the subgraph induced by the vertices of inrank~$0$. It remains to show that there exists such a coloring that is optimal.

  Let~$\angles{V_1,\dots,V_k}$ be the color classes of an optimal coloring of~$G$. The set~$V_1$ must be independent and cannot contain vertices of inrank greater than~$0$, as these vertices are colored with the color~$1$ and can therefore have no incoming arcs. Let~$V_1'\supseteq V_1$ be a maximal independent set in the subgraph induced by the vertices of inrank~$0$. As~$G-V_1'\subseteq G-V_1$ it follows that~$G-V_1'$ is~$(k-1)$-colorable and thus that~$\chi(G)\geq 1+\chi(G-V_1')$. 
\end{claimproof}

\MISZBound*
\label{claim:mis0-bound*}
\begin{claimproof}
  Let~$V_0$ be the set of vertices of inrank~$0$. It holds that~$\ndu(G[V_0])\leq \ndu(G)$ and $\omega(G[V_0])\leq \omega(G)$. For each type~$I$ that induces an independent set, it holds that either all or none of the vertices of~$I$ are contained in a maximal independent set of~$G[V_0]$, as all vertices in~$I$ have the same neighborhood. Furthermore, for each type~$C$ that induces a clique, it holds that at most one vertex of~$C$ is contained in a maximal independent set of~$G[V_0]$. Thus, as~$\omega(G)\geq 1$, it holds for each of the at most~$\ndu(G)$ types that there are at most~$\omega(G)+1$ possibilities to choose vertices of the type for a maximal independent set, resulting in at most~$(\omega(G)+1)^{\ndu(G)}$ maximal independent sets of~$G[V_0]$ in total.
\end{claimproof}

\end{document}